\documentclass[a4paper,UKenglish,cleveref, autoref, thm-restate]{lipics-v2021}

\title{Perpetual Exploration of a Ring in Presence of Byzantine Black Hole} 

\author{Pritam Goswami}{Sister Nivedita University, West Bengal, India}{pgoswami.cs@gmail.com}{[OrcidId:0000-0002-0546-3894]}{}

\author{Adri {Bhattacharya}}{Indian Institute of Technology Guwahati, Assam, India}{a.bhattacharya@iitg.ac.in}{[OrcidId:0000-0003-1517-8779]}{}

\author{Raja Das}{Jadavpur University, West Bengal, India}{rajad.math.rs@jadavpuruniversity.in}{[OrcidId:0009-0002-0161-5341]}{}

\author{Partha Sarathi Mandal}{Indian Institute of Technology Guwahati, Assam, India}{psm@iitg.ac.in}{0000-0002-8632-5767}{}

\titlerunning{Goswami et al.} 

\Copyright{Jane Open Access and Joan R. Public} 



\category{} 

\relatedversion{} 

\acknowledgements{I want to thank \dots}

\EventEditors{Silvia Bonomi and Etienne Rivière}
\EventNoEds{2}
\EventLongTitle{Conference on Principles of Distributed Systems (OPODIS 2024)}
\EventShortTitle{OPODIS 2024}
\EventAcronym{OPODIS}
\EventYear{2024}
\EventDate{December 11--13, 2024}
\EventLocation{Lucca, Italy}
\EventLogo{}
\SeriesVolume{}
\ArticleNo{}

\usepackage{amsfonts, amssymb, amsmath,pifont,float,caption,multicol}
\usepackage[ruled,vlined,linesnumbered]{algorithm2e}
\usepackage[symbol]{footmisc}
\usepackage{thmtools} 
\usepackage{thm-restate}
\nolinenumbers
\begin{document}

\maketitle

\begin{abstract}
Perpetual exploration stands as a fundamental problem in the domain of distributed mobile agent algorithms, where the objective is to ensure that each node within a graph is visited by at least one agent infinitely often. While this issue has received significant attention, particularly concerning ring topologies, the presence of malicious nodes, referred to as black holes, adds more complexity. A black hole can destroy any incoming agent without leaving any trace of its existence.
 
In \cite{BampasImprovedPeriodicDataRetrieval,KralovivcPeriodicDataRetrievalFirst}, the authors have considered this problem in the context of \textit{ periodic data retrieval}. They introduced a variant of black hole called gray hole (where the adversary chooses whether to destroy an agent or let it pass) among others, and showed that 4 asynchronous and co-located agents are necessary and sufficient to solve the periodic data retrieval problem (hence perpetual exploration) in the presence of such a gray hole if each of the nodes of the ring has a whiteboard.
 
 This paper investigates the exploration of a ring by introducing a realistic variant of gray hole, termed as a ``Byzantine black hole''. In addition to the capabilities of a gray hole, the adversary can also choose whether to erase any previously stored information on that node. 

 Note that in \cite{BampasImprovedPeriodicDataRetrieval,KralovivcPeriodicDataRetrievalFirst}, the authors only considered one particular initial scenario (i.e., agents are initially co-located) and one specific communication model (i.e., whiteboard). Now, there can be many other initial scenarios where all agents might not be co-located (i.e., they may be scattered). Also, there are many weaker communications models such as \textit{Face-to-Face, Pebble} where this perpetual exploration problem is yet to be investigated, in presence of a Byzantine black hole. 
  
 The main results of our paper emphasize minimizing the number of agents while guaranteeing that they perform the perpetual exploration on a ring even in the presence of a Byzantine black hole under different communication models and for different starting scenarios. On the positive side, as a byproduct of our work, we achieved a better upper and lower bound result (i.e., 3 agents) for perpetual exploration in the presence of a Byzantine black hole (which is a more generalized version of a gray hole), by trading-off the scheduler capability, when the agents are initially co-located, and each node contains a whiteboard.
 \end{abstract}
\vspace{-0.45cm}
\section{Introduction}
\vspace{-0.1cm}
Exploring a set of nodes in a network is one of the fundamental tasks in the domain of distributed computing by mobile agents, formulated in the year of 1951 by Shannon \cite{Shannon-First-Exploration}. Now, the security of these mobile agents while exploring these networks is one of the important issues that need to be addressed. Among all the possible security threats that are addressed yet in literature, two among them are the most prominent, i.e., the threats from a \textit{malicious agent} \cite{IntrusionDetection} and the threats from a \textit{malicious host} \cite{DobrevAnonymousRingBHS}. In this paper, we are interested in the latter case, where the threats are from a malicious host. This host is a stationary node in the network, which has the ability to destroy any incoming agents without leaving any trace of its existence. So, the first task of the mobile agents operating in the network must be to locate this malicious node. Note that the most trivial optimisation parameter to ensure while locating this malicious host (also termed as a black hole), is that a minimum number of agents gets destroyed by this node. This problem of locating the black hole by mobile agents is termed as \textit{black hole search} (also termed as BHS problem) problem. BHS problem is studied from the year 2006, when Dobrev et al. \cite{DobrevSeminalPaperBHS} first introduced it. After which, till date, there have been many variations to this problem, some of them are \cite{CzyzowiczComplexityBHS,CzyzowiczBHSSyncTree,DobrevDangerousGraphTokens,DobrevAnonymousRingBHS,DobrevBHSUn-orientedRingScattered}. This problem has various real life implications, such as the black hole can be a virus in the network or it can be some crash failure, such that this node resembles the characteristic of a black hole, after the failure.

Observe that, to detect the black hole, there needs to be some agent that has to visit that particular node. Further, since any agent visiting the node gets destroyed, there must be some communication tool, that can render this information to other alive agents, such that at least one agent remains alive, knowing the location of the black hole. Three such communication tools have been predominantly used in literature: a) \textit{whiteboard} model \cite{DobrevBHSAnonymousRingWhiteboard}, in which there is a storage capacity at each node, which an agent can use to leave a message by reading its contents and writing some new information, b) \textit{pebble} model \cite{FlocchiniPingPongBHSPebble}, an agent can carry a movable token from one node to another, c) \textit{face-to-face} model \cite{DiLunaBHSDynamicRing}, in this case, an agent can share and communicate with another agent when they are at the same node and at the same time. In addition to the communication tools, the initial locations of the agents (i.e., whether the agents are initially scattered \cite{DobrevBHSUn-orientedRingScattered} or they are co-located \cite{DiLunaBHSDynamicRing}) is also one of the important parameters, generally studied in the literature. 

Further, the most studied version of a black hole has a fairly basic nature, i.e., only destroying any incoming agent. Note that, in reality black holes may not be so simple; they may have many ways to disrupt the movement or harm an agent. Considering this phenomenon, we in this paper have tried to consider a black hole that has more capabilities other than just destroying any incoming agent. In our case a black hole may or may not kill any incoming agent; it may do so based on an adversary, which decides when to destroy an incoming agent and when not to. Whenever it decides not to destroy an agent, it simply behaves like any other node in the network, disguising it from the rest of the nodes, and creating no anomaly for the visiting agent. In addition to this, we have also considered that the black hole has further capabilities; it can also choose whether to destroy the message (i.e., stored data in case of a whiteboard, and place the token in case of pebble) at that node along with the incoming agent. This choice is also maintained by an adversary as well. We call this kind of black hole a \textit{Byzantine black hole}.

Our aim in this paper is to solve the problem of perpetual exploration in a network, i.e., visiting every node in the network infinitely often other than the Byzantine black hole, by the mobile agents. Previously, in \cite{BampasImprovedPeriodicDataRetrieval,KralovivcPeriodicDataRetrievalFirst} the authors introduced a set of models for a black hole, which has more capabilities other than just destroying an agent, they term these malicious nodes as \textit{gray} hole, \textit{gray$^{+}$} hole and \textit{red} hole, respectively based on their capabilities. They considered the following characteristics: they can fake agents, change the whiteboard contents, change the ports different from the requested ones, or can change the FIFO ordering as well.  In this context, they solved perpetual exploration in a ring (which they term as \textit{periodic data retrieval} problem) by a team of asynchronous mobile agents, under these aforementioned various black hole characteristics, in which the agents are initially co-located and each node in a network has a whiteboard. On the other hand, the results of above-mentioned papers only holds for the case, when initially the agents are co-located and each node of the ring has a whiteboard. Note that there can be many other initial positions for the agent to start with, and also, whiteboard is a very powerful communication tool in this domain of study of mobile agents. So, in this paper, we investigate these gaps in the context of perpetual exploration, considering the presence of one Byzantine black hole using synchronous agents. Also note that the position of a Byzantine black hole can be arbitrary (except the starting locations of the agent) but fixed. The Byzantine black hole we considered is a generalized version of the gray hole, as it also can choose whether to erase any previous data stored at that node, the moment it acts as a black hole. 
\vspace{-0.35cm}
\subsection{Related Works}
\vspace{-0.1cm}
The black hole search (i.e., BHS) problem is a prominent variation of exploration problem studied in the literature. A survey of which can be found in \cite{PengBHSSurvey}. This problem is investigated under various topologies (such as trees \cite{CzyzowiczBHSSyncTree}, rings \cite{DobrevAnonymousRingBHS}, tori \cite{ChalopinBHSScatteredTorus} and in arbitrary and unknown networks \cite{CzyzowiczComplexityBHS,DobrevDangerousGraphTokens}). All these discussed networks are static in nature. Recently, there has been a lot of interest in dynamic networks. The following papers \cite{AdriBHSDynCactus,AdriBHSDynTori,DiLunaBHSDynamicRing}, studied the BHS problem on dynamic ring, dynamic torus and dynamic cactus graph, where the underlying condition is that, irrespective of how many edges are dynamic in nature, the network must remain connected at any time interval (which is also termed as \textit{1-interval connected}). In rings, the BHS problem has been studied for different variants; the most predominant among them are choice of schedulers (i.e., synchronous \cite{ChalopinBHSSyncScattRing} and asynchronous \cite{BalamohanBHSAsyncRing}), communication tools (i.e., face-to-face \cite{CzyzowiczComplexityBHS}, pebble \cite{FlocchiniPingPongBHSPebble} and whiteboard \cite{BalamohanBHSAsyncRing}) and initial position of the agents (i.e., co-located \cite{BalamohanBHSAsyncRing} and scattered \cite{ChalopinBHSSyncScattRing}).

The most relevant papers related to our work are the papers by Královič et al. \cite{KralovivcPeriodicDataRetrievalFirst} and by Bampas et al. \cite{BampasImprovedPeriodicDataRetrieval}.  The paper by Královič et al. \cite{KralovivcPeriodicDataRetrievalFirst} is the first to introduce a variant of this black hole, where the black hole has the ability to either choose to destroy an agent or let it pass (which they term as \textit{gray hole}). Further they extended the notion of gray hole, where the gray hole has the following additional capabilities: it has the ability to alter the run time environment (i.e., changing the whiteboard information), or it has the ability to not to maintain communication protocol (i.e., do not maintain FIFO order). They solved this problem under an asynchronous scheduler on a ring, only when the agents are initially co-located and each node in the network has a whiteboard. The following results are obtained by them, they gave an upper bound of 9 agents for performing periodic data retrieval (i.e., which is equivalent to perpetual exploration) in the presence of a gray hole; further, in addition to gray hole, when the whiteboard is unreliable as well, they proposed an upper bound of 27 agents. Next, Bampas et al. \cite{BampasImprovedPeriodicDataRetrieval} significantly improved the earlier results. They showed a non-trivial lower bound of 4 agents and 5 agents for gray hole case and for gray hole with unreliable whiteboard case, respectively. Further, with 4 agents as well, they obtained an optimal result for the gray hole case, whereas with 7 agents proposed a protocol for the case with gray hole and unreliable whiteboard. As far as we are aware, we are the first to investigate the perpetual exploration problem of a ring under different communication tools (i.e., face-to-face, pebble and whiteboard) as well as for different initial positions (i.e., co-located and scattered), for a variant of gray hole, 
where it can erase any previously stored information but can not alter it. We term this type of gray hole a Byzantine black hole. In the following part, we discuss the results we have obtained.\\
\noindent\textbf{Our Contribution:} In this paper, we investigate the perpetual exploration problem, by a team of synchronous mobile agents, of a ring $R$ of size $n$, in the presence of a \textit{Byzantine black hole}. First, we consider the case when the agents are initially co-located. We obtain the following results.\\
\noindent\textbf{A:} For \textit{Pebble} model of communication, we obtain that 3 agents are necessary and sufficient to perpetually explore $R$.\\
\noindent\textbf{B:} For \textit{Face-to-Face} model of communication, we obtain that 5 agents are sufficient to perpetually explore $R$.\\
\noindent\textbf{C:} For \textit{Whiteboard} model as well, we achieve the same lower and upper bounds as mentioned in \textbf{A}. This result shows that, by considering the scheduler to be synchronous instead of asynchronous (as assumed in \cite{BampasImprovedPeriodicDataRetrieval}), the tight bound on the number of agents to perpetually explore $R$, reduces from 4 to 3.\\  
Next, we consider the case, when the agents are initially scattered, and in this context, we obtain the following results:\\
\noindent\textbf{D:} For \textit{Pebble} model of communication, we show that 4 agents are necessary and sufficient to explore $R$ perpetually .\\
\noindent\textbf{E:} For \textit{Whiteboard} model of communication, we obtain an improved bound of 3 agents (in comparison to \textbf{D}), which is necessary and sufficient to explore the ring $R$ perpetually.\\
In the following Table \ref{tab:cross_table}, we have summarized the results.



\begin{table}[h]
\centering
\begin{tabular}{|c|c|c|c|c|}
\hline
& &  \textbf{Whiteboard} & \textbf{Pebble} & \textbf{Face-to-Face} \\
\hline
\multirow{2}{*}{\textbf{Co-located}}& Upper Bound & 3  & 3 & 5  \\
\cline{2-5}
 & Lower Bound & 3 & 3 & 3\\
\hline
\multirow{2}{*}{\textbf{Scattered}}& Upper Bound & 3  & 4 & --- \\
\cline{2-5}
 & Lower Bound & 3 & 4 & Non-Constant \cite{DiLunaBHSDynScattRing}\\
\hline
\end{tabular}
 \caption{Summary of our results}
\label{tab:cross_table}
\end{table}
\noindent\textbf{Organisation:} Rest of the paper is organised as follows. Section \ref{model} presents model and preliminaries. In Section \ref{section: impossibility results}, we discuss some impossibility results. In Section \ref{section: co-located agents} and \ref{section: scattered agents}, we propose perpetual exploration algorithms for the agents, when the agents are initially co-located and scattered, respectively. Lastly, we conclude in Section \ref{section: conclusion}. Note that due to page restriction, detailed description of the algorithms and proofs of theorems and lemmas are referred to the \textbf{Appendix}. 
\section{Model and Preliminaries}\label{model}
\vspace{-0.1cm}
In this paper, we consider the underlying topology of the network as an oriented ring $R = \{v_0, v_1, \dots, v_{n-1}\}$. Each node $v_i$ (where $i \in \{0,1,\dots,n-1\}$) is unlabeled and has two ports connecting $v_{(i-1)\bmod{n}}$ and $v_{(i+1)\bmod{n}}$, labeled consistently as \textit{left} and \textit{right}. A set $A = \{a_0, a_1, \dots, a_{k-1}\}$ of $k$ agents operates in $R$. We consider two types of initial positions for the set $A$ of agents. In the first type, each agent in $A$ is \textit{co-located} at a node, which we term as their \textit{home}. In the second type, the agents can start from several distinct nodes, which we term as \textit{scattered} initial positions. Each agent has knowledge of the underlying topology $R$ and possesses some computational capabilities, thus requiring $\mathcal{O}(\log n)$ bits of internal memory. The agents have unique IDs of size $\mathcal{O}(\log k )$ bits taken from the set $[0,k^c]$ ($c$ is a constant), which are perceived by other agents when they are co-located. The agents are autonomous and execute the same set of rules (i.e., they execute the same algorithm).

We consider three types of communication that the agents have in order to communicate with other agents. \textit{Face-to-Face} (F2F): In this model, an agent can communicate with another agent when they are co-located.
    \textit{Pebble}: In this model, the agents are equipped with a movable token (also termed as ``pebble''), which signifies a single bit of information. The agents can carry a pebble from one node in a fairly mutually exclusive way and also can drop it on any other node. Agents use this pebble to mark some special nodes, for other agents to distinguish it.
    \textit{Whiteboard}: In this case, each node of $R$ contains $\mathcal{O}(\log n)$ bits of memory, which can be used to store and maintain information. Any agent can read the existing information or write any new information on the whiteboard of its current node. Note that fair mutual exclusion is maintained, i.e., concurrent access to the whiteboard data is not permitted.

The agents operate in synchronous rounds, and in each round, every agent becomes active and takes a local snapshot of its surroundings. For an agent at a node $v$ in some round $r$, the snapshot contains two ports incident to $v$, already stored data on the memory of $v$ (if any, only in the case of the whiteboard model of communication), the number of pebbles located at $v$ (if any, only in the case of the pebble model of communication), contents from its local memory, and IDs of other agents on $v$. Based on this snapshot, an agent executes some action. This action includes a \textit{communication} step and a \textit{move} step. In a communication step, an agent can communicate implicitly or explicitly with other agents according to the communication models discussed above. In the move step, the agent can move to a neighbouring node by following a port incident to $v$. Thus, if an agent at node $v$ in round $r$ decides to move during the move step, in round $r+1$, it resides on a neighbour node of $v$. All these actions are atomic, so an agent cannot distinguish another agent concurrently passing through the same edge; instead, it can only interact with another agent (based on its communication model) when it reaches another node.

A \textit{black hole} is a stationary malicious node in an underlying graph, which has the ability to destroy any visiting agent without leaving any trace of its existence. Furthermore, the black hole nature of the node is controlled by an adversary. In addition to this, whenever the black hole nature is activated, the adversary can also choose to destroy any information stored at that node. We term this kind of node a \textit{Byzantine Black Hole}.

In this paper, we assume that the underlying graph contains a single Byzantine black hole, while the other nodes are normal nodes, termed as \textit{safe nodes}. It is assumed that the starting positions of each agent must be a safe node. This Byzantine black hole node is unknown to the agent. Here, we assume that if the adversary decides to activate the black hole nature of the Byzantine black hole node, it does so at the beginning of its corresponding round, and the node retains that nature until the end of this current round. Furthermore, we have considered that our Byzantine black hole has the ability always to destroy any incoming agent during its black hole nature and also choose to destroy any information present on that node. This paper aims to perpetually explore the ring $R$ with the minimum number of agents. Next, we formally define our problem:

\begin{definition}[\textsc{PerpExploration-BBH}]
Given a ring network $R$ with $n$ nodes, where one node ($v_b$) is a Byzantine black hole, and with a set of agents $A$ positioned on $R$, the \textsc{PerpExploration-BBH}, asks the agents in $A$ to move in such a way that each node of $R$, except $v_b$, is visited by at least one agent infinitely often.
\end{definition}
\vspace{-0.45cm}
\section{Impossibility Results}\label{section: impossibility results}
\vspace{-0.1cm}
Here, at beginning, we state the first impossibility result which gives us a lower bound on minimum number of agents required to solve \textsc{PerpExploration-BBH}.
\begin{theorem}
\label{impossible:whitebrd2}
    A set of two synchronous agents in a ring $R$ of size $n$ cannot solve \textsc{PerpExploration-BBH}, even in the presence of a whiteboard if number of possible consecutive black hole positions is at least 3.
\end{theorem}
The main idea of the proof of the above Theorem~\ref{impossible:whitebrd2} is that we create a scenario using adversarial techniques, where after one agent is destroyed by the Byzantine black hole $v_b$ (say), the other and only alive agent ends up with at least two possible choices for the position of the black hole, which in turn creates confusion for the agent. Hence it is unable to correctly detect the black hole node among these two sets of nodes. In this scenario it is impossible for the only alive agent to successfully explore the whole ring except $v_b$ perpetually, either without correctly detecting the Byzantine black hole position, or without getting destroyed as well. As a direct implication of Theorem~\ref{impossible:whitebrd2}, we can have the following corollaries. (Please see \textbf{Appendix} in Section \ref{Appendix:thm 3} for the detailed proof)

\begin{corollary}
\label{cor: necessity 3 agnets whitbrd}
    A set of three synchronous agents are required to solve the \textsc{PerpExploration-BBH} on a ring $R$ with $n$ nodes, where each node of $R$ has a whiteboard.
\end{corollary}
Note that this lower bound of 3 agents is also true for agents with the pebble model of communication, as the pebble model of communication can be easily simulated in the whiteboard communication model.
\begin{corollary}
\label{impossible: pebble2}
A set of two agents, each equipped with $O(\log n)$ pebbles, can not solve the \textsc{PerpExploration-BBH} problem on a ring $R$ with $n$ nodes.
\end{corollary}
 Our next result further improves the lower bound for number of agents, when agents are scattered and have pebble model of communication.
\begin{theorem}
\label{impo:pblscat}
    A set of 3 scattered agents, each equipped with a pebble, can not solve the \textsc{PerpExploration-BBH} problem on a ring $R$ with $n$ nodes.
\end{theorem}
We have proved Theorem~\ref{impo:pblscat} using an adversarial technique. We created an instance of three copies of the same ring $R$ (i.e., $R_1$, $R_2$ and $R_3$) in such a way that the position of Byzantine black hole ($v_b)$ are different in each of them but they are consecutive in $R$. Also, in each case, the agents are initially placed at the nodes of $R$, which are at a distance of equal length between each other. Now, when an agent is destroyed by $v_b$, the other two agents can not distinguish between these three copies (due to same information on nodes for each of them). So, the problem reduces to two agents with 3 pebbles solving \textsc{PerpExploration-BBH} where a number of possible consecutive black hole positions is at least three. Then by Corollary~\ref{impossible: pebble2}, we can have the desired result (For detailed proof, see \textbf{Appendix} in Section \ref{Appendix: thm 6}). Next, corollary is a direct consequence of Theorem \ref{impo:pblscat}. 
\begin{corollary}
\label{cor:necessity of 4 agents scat pbl}
    A set of four scattered agents, each with a pebble, are required to solve \textsc{PerpExploration-BBH} on a ring $R$ with $n$ nodes.
\end{corollary}
\vspace{-0.45cm}

\vspace{-0.2cm}
\section{Perpetual Exploration with Co-located Agents}\label{section: co-located agents}
\vspace{-0.1cm}
In this section we assume that initially agents are co-located at a node which is termed as \textit{home}. On the basis of this assumption here we investigate the sufficiency for the number of agents to solve \textsc{PerpExploration-BBH} problem under different model of communication.
\vspace{-0.25cm}
\subsection{Pebble Model of Communication}
\vspace{-0.1cm}
In this section, we consider the communication model, where the starting node has $k-1$ identical and movable tokens (termed as pebbles), where $k$ is the total number of agents deployed. A pebble can be carried by an agent from one node to another. This pebble acts as a mode of communication for the agents, as the agents can perceive the presence of a pebble at the current node. Moreover, an agent can also perceive the presence of other agents which are co-located at the current round (i.e., gather the IDs of the other co-located agents). Our main aim is to prove the following theorem in this section.
\begin{theorem}
    \label{Thm: 3 coloc pbl necessary and sufficient}
    A team of 3 co-located, synchronous agents are necessary and sufficient to solve \textsc{PerpExploration-BBH} on a ring $R$, with $n$ nodes under the pebble model of communication,with the presence of two pebbles initially co-located with the agents. 
\end{theorem}

The necessary part follows from Corollary~\ref{impossible: pebble2}. For the sufficiency part we present an algorithm, \textsc{PerpExplore-Coloc-Pbl}. We have provided a very brief description of the algorithm. For detailed description and pseudo code see \textbf{Appendix} (Section~\ref{appendix: colocpbl}).

\noindent\textbf{Brief Description of the Algorithm:} Given a ring $R$ with a Byzantine black hole node $v_b$, a set of agents, $A=\{a_1,a_2,a_3\}$ (where we assume, ID of $a_1~<$ ID of $a_2~<$ ID of $a_3$), are initially co-located at $home$. The agents have no knowledge about the position of the node $v_b$, the only set of knowledge the agents have are: the total number of nodes in the underlying topology $R$ (i.e., $n$), and also the knowledge that the initial node, i.e., $home$ is safe. So, the remaining arc of $n-1$ nodes in $R$ is a suspicious region (which we term as $S$, where the cardinality of $S$ i.e., $|S|$ denotes the length or size of the suspicious region) for each of the three agents. 

The  main idea of this algorithm is as follows: initially two agents $a_1$ and $a_2$ explores the ring perpetually, while $a_3$ waits at the $home$. $a_1$ and $a_2$ explores the ring $R$, executing the rule described as follows. If all the three agents are at $home$ along with two pebbles, then in the $r-$th round ($r \ge 0$), $a_1$ moves clockwise with one pebble and it waits at round $r+1$ for $a_2$ to arrive. In round $(r+1)$, $a_2$, leaving the remaining pebble at $home$ moves to the next node along clockwise direction to meet $a_1$. In the subsequent round, i.e., $(r+2)-$th round, $a_1$ after meeting $a_2$ moves again to the next node in the clockwise direction and waits for 3 rounds (i.e., till $(r+5)-$th round) for $a_2$ 
if and only if $a_1$ is not at $home$. In the mean time at $(r+3)-$th round, $a_2$ leaves its current node, moves one step in counter-clockwise direction, collects the left behind pebble. Then in $(r+4)-$th round $a_2$ moves again to the earlier node in clockwise direction. Subsequently, in $(r+5)-$th round, $a_2$ again leaves the pebble at its current node, and moves in a clockwise direction, to meet $a_1$ (which is waiting at the current node for $a_2$ to arrive), for better reference see Fig. \ref{fig:enter-label}. Note that, when $a_2$ meets $a_1$ outside $home$, the pebble it was carrying is left by $a_2$ at the nearest counter-clockwise node of the meeting node (i.e., where $a_2$ meets with $a_1$).
 Now from round $(r+6)$ onwards to $(r+9)-$th round, the same execution repeats, as it has occurred between round $(r+2)$ to round $(r+5)$. The only difference is that, if both agents (i.e., $a_1$ and $a_2$) at round $r+2$ were at a node $v_1$, at round $r+6$, they are on the nearest clockwise node of $v_1$ (i.e., at $v_2$, say).  This process continues until $a_1$ reaches $home$. In this case $a_1$ does not move clockwise, until $a_2$ meets it with the pebble it was carrying, in that case, it waits further until $a_2$ brings that pebble back at $home$. For this to happen after reaching $home$, 5 rounds are sufficient for waiting. When $a_2$ reaches $home$ with the pebble, then all 3 agents are again at the $home$ with 2 pebbles. This way the 3 agents explore the ring. The exploration can hamper if either $a_1$ or, $a_2$, or a pebble is destroyed by the Byzantine black hole $v_b$ while the above mentioned procedure is executed by them.

 \begin{figure}
     \centering
     \includegraphics[scale=0.4]{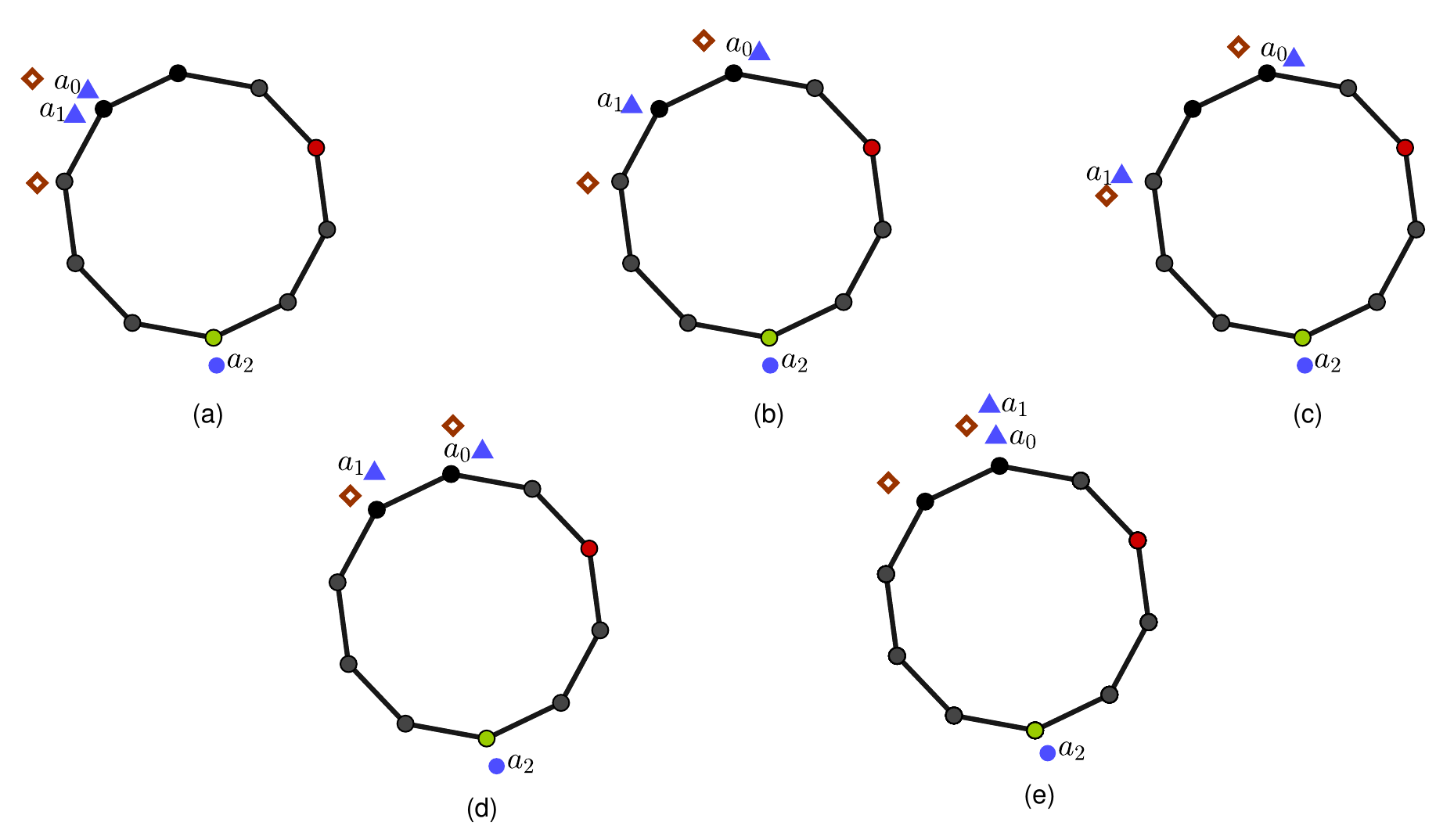}
     \caption{An execution of \textsc{PerpExplore-Coloc-Pbl}, starting from the configuration where $a_0$ and $a_1$ are together on a vertex, to the configuration where $a_0$ and $a_1$ are on the same vertex again, which is the clockwise neighbour of the earlier vertex. The red node is the byzantine black hole, the green node is the $home$ and red boxes are pebbles.}
     \label{fig:enter-label}
 \end{figure}
 
 Firstly, let only $a_1$ is destroyed at some round $t$. Then there must exists a round $t'<t$, when $a_2$ was with $a_1$ at some node $v_0$ (say) for the last time. This means at round $t'$, $a_1$ moved clockwise to the next node $v_1$ along with the pebble it was carrying, and it must have been destroyed before or at the round when $a_2$ reaches $v_1$ again (in this case $v_1$ is $v_b$). Now when $a_2$ reaches $v_1$ there can be two cases. Either $a_2$ remains alive or (as the adversary may choose not to destroy $a_2$), it can get destroyed as well. For the initial case, $a_2$ identifies the Byzantine black hole by not seeing $a_1$ there and then it can leave that node and can explore the ring perpetually avoiding the node. The latter case is equivalent to both $a_1$ and $a_2$ are destroyed by the Byzantine black hole $v_b=v_1$. Note that since outside $home$ whenever $a_2$ reaches the same node of $a_1$, it leaves behind a pebble at the nearest counter-clockwise node (here $v_0$). For this case, another agent (i.e., $a_3$) which was waiting at $home$, finds none of the two exploring agents returns at $home$, even after waiting for a sufficient number of rounds. This incident triggers $a_3$ to move clockwise until it finds a pebble (one that is left behind by $a_2$ at $v_0$). Whenever the pebble is found, $a_3$ knows that the next clockwise node is the Byzantine black hole and it starts exploring $R$ avoiding that node.
 
 Now, there can be two other cases that can hamper the above mentioned exploring procedure. 
 For the first case, let only $a_2$ is destroyed at some round $t>0$. Let $t'<t$ be the last round before $t$ when both $a_1$ and $a_2$ were together at a node $v_0$. At round $t'$, $a_1$, moves clockwise to the next node $v_1$. Now if $a_2$ is not destroyed in the next 3 rounds, then it meets $a_1$ at $v_1$ again, contrary to our assumption. So, if only $a_2$ is destroyed then it must be at any of the rounds $t'+1, t'+2$ or $t'+3$. In this 3 round $a_2$ can be either on $v_0$ or $v_{-1}$ where, $v_{-1}$ is the counter-clockwise nearest node of $v_0$. In this case, when $a_1$ finds $a_2$ has not arrived at $v_1$ even after waiting for 3 rounds, in which case, it understands that either $v_0$ or $v_{-1}$ is the Byzantine black hole. This knowledge triggers $a_1$ to move clockwise until it meets $a_3$ at $home$, leaving the pebble along with it behind at $v_1$. When $a_3$ sees that even after certain round of waiting only $a_1$ is able to arrive at $home$, and that too without the pebble it was supposed to be carrying, $a_3$ understands that $a_1$ must have detected an anomaly. In this case both $a_1$ and $a_3$ moves counter-clockwise together to find the pebble left by $a_1$. When they find it they declare $v_1$ as $home$. After this they start exploring the ring in different directions in the following way. $a_1$ moves clockwise until $v_{-1}$ and $a_3$ starts moving counter-clockwise until $v_0$. After that, they both move in opposite direction until they again reach $home$ (which is the new $home$) and waits for sufficient number of rounds to meet each other. This way the exploration keeps going on until one among $a_1$ or $a_3$ is destroyed. Let $a_3$ fails to reach $home$. In that case, $a_1$ will detect that $v_0$ must be the Byzantine black hole and it can start exploring the ring $R$ avoiding $v_0$. On the other hand if $a_1$ fails to reach then it must be destroyed by the Byzantine black hole which is nothing but the node $v_{-1}$. Knowing this $a_3$ then explores the ring avoiding that node. Now there can be another case when no agents but a pebble is destroyed at $v_b$. This pebble must be the pebble that is left behind at $v_{-1}$ by the agent $a_2$ when it reaches $v_0$ to meet $a_1$. In this case, only the pebble can be destroyed before $a_2$ reaches to collect the pebble back. So, when $a_2$ reaches $v_{-1}$, it finds out that there is no pebble. This leads to $a_2$, knowing that $v_{-1}$ is the Byzantine black hole and in which case it starts exploring the ring $R$ avoiding that node. If $a_2$ is also destroyed while it reaches $v_{-1}$ to collect the pebble, this case is similar to the case where we described the algorithm when only $a_2$ is destroyed. 

 \noindent\textbf{Sketch of Correctness:}
 To prove the correctness of algorithm \textsc{PerpExplore-Coloc-Pbl} we prove the following Theorem. Here we just present a sketch of the proof, for details see \textbf{Appendix} (Section~\ref{Appedix: coloc pbl correct}).
 \begin{theorem}
 \label{thm:colocPblCorrect}
 Algorithm \textsc{PerpExplore-Coloc-Pbl} solves \textsc{PerpExploration-BBH} on a ring $R$ with 3 co-located and synchronous agents under the pebble model of communication with the presence of two pebbles initially co-located with the agents.
\end{theorem}
We first prove that if no agents are destroyed then either the agents continue with the exploration of ring $R$ without knowing the exact location of the Byzantine black hole, or there exists at least one agent that knows the exact location of the Byzantine black hole which can then explore the ring indefinitely avoiding the Byzantine black hole (See Lemma~\ref{lemma: either ring is explored or BBH is detected} in \textbf{Appendix}). Next we proved that, if one agent is destroyed while exploring the ring, then either there exists one agent that identifies the Byzantine black hole uniquely and continues to indefinitely explore all nodes of the ring, except the Byzantine black hole or the exploration continues while the length of the suspicious region  (i.e., $|S|$) decreases to 2 from $n-1$ (Lemma~\ref{lemma: exactly one destoyed} in \textbf{Appendix}). For the later case we proved that the exploration continues until another agent is destroyed by the Byzantine black hole (Remark~\ref{Remark: exploration with two agents} in \textbf{Appendix}). In this case when two agents are destroyed by the Byzantine black hole, the only alive agent can uniquely identify the Byzantine black hole without moving on to it (Lemma~\ref{lemma: exactly two agents destroyed} in \textbf{Appendix}). Thus it can then perform the exploration of $R$ avoiding the Byzantine black hole. This proves Theorem~\ref{thm:colocPblCorrect}.

\subsection{Face-to-Face Model of Communication}

In this section, we consider the Face-to-Face (also termed as F2F) model and prove the following theorem.

\begin{theorem}\label{thm:F2FCorrectness}
    A team of 5 co-located and synchronous agents are sufficient to solve \textsc{PerpExploration-BBH} on a ring $R$ with $n$ nodes under the F2F model of communication.
\end{theorem}

In order to prove Theorem \ref{thm:F2FCorrectness}, we discuss an algorithm with 5 co-located agents, where each agent can communicate among themselves at the same node at the same round (i.e., each agent have F2F model of communication). Initially 3 lowest ID agents, say, $a_1$, $a_2$ and $a_3$ are chosen, where fourth lowest ID agent say, $a_4$, is associated with $a_1$ and the largest ID agent, say $a_5$, is associated with $a_2$. The idea of the algorithm resembles to that of the algorithm \textsc{PerpExplore-Coloc-Pbl}. Note that as in this case there is no existence of pebbles, hence, we configure the behaviour of the agents $a_4$ and $a_5$ in such a way that they act as pebbles, associated with $a_1$ and $a_2$ in \textsc{PerpExplore-Coloc-Pbl}, respectively. In \textsc{PerpExplore-Coloc-Pbl} when we say that an agent $a_i$ carries a pebble $p$, in this case we mean that the agent $a_i$ communicates a message \texttt{carry} with its associated agent $a_{i'}$ such that both these agents simultaneously move together until further new instruction is communicated. On the contrary as per \textsc{PerpExplore-Coloc-Pbl} when we say that an agent $a_i$ drops a pebble $p$ at a node $v$, then in this case we mean that $a_i$ communicates a message \texttt{drop} with $a_{i'}$, which in turn instructs $a_{i'}$ not to move further and remain stationary at the node $v$ until further instruction is communicated. Hence, with this terminology, a team of 5 agents can execute the algorithm \textsc{PerpExplore-Coloc-Pbl}, and the correctness also follows similarly. This shows that a team of 5 co-located and synchronous agents are sufficient to solve \textsc{PerpExploration-BBH} under F2F model of communication. This proves Theorem \ref{thm:F2FCorrectness}.

\subsection{Whiteboard Model of Communication}
\label{subsection: colocwhitbrd simul}
\vspace{-0.1cm}
The algorithm \textsc{PerpExplore-Coloc-Pbl} can be simulated in whiteboard model of communication. A pebble on a node can be simulated by a bit of information, which is marked on the whiteboard of that node. Note that, collecting a pebble from the node is simulated by erasing the same bit of information, on that node and dropping the pebble can be simulated by marking the node with a bit of information as well on the whiteboard of that node. From this we get the following result for whiteboard model of communication. 
\begin{theorem}
\label{thm:colocwhitbrd}
     A team of 3 synchronous co-located agents are necessary and sufficient to solve \textsc{PerpExploration-BBH} on a ring $R$ if each node are equipped with a whiteboard having constant memory.
\end{theorem}
\vspace{-0.2cm}
Previously in \cite{BampasImprovedPeriodicDataRetrieval}, for asynchronous scheduler the tight bound for number of agents to solve this problem was 4. Now from Theorem~\ref{thm:colocwhitbrd}, we get a trade-off between reducing the optimal number of agents required vs the scheduler, in presence of a more generalized version of gray hole as well (i.e., Byzantine black hole). 
\vspace{-0.45cm}
 \section{Perpetual Exploration with Scattered Agents}\label{section: scattered agents}
 \vspace{-0.1cm}
 Here in this section, we discuss the problem of \textsc{PerpExploration-BBH}, when the agents are initially scattered on more than one nodes, where each of these starting nodes are assumed to be safe. We investigate this problem, under different communication models. First, we discuss the pebble model of communication and subsequently we discuss the whiteboard model of communication.
 \vspace{-0.35cm}
 \subsection{Pebble Model of Communication}
 \vspace{-0.1cm}
 Our goal in this section is to prove the following theorem.
 \begin{theorem}
     A team of 4 synchronous scattered agents with one pebble each is necessary and sufficient to solve the \textsc{PerpExploration-BBH} problem on a ring $R$ with $n$  nodes, when the agents are initially scattered on $R$.
 \end{theorem}

 The necessary part follows from Corollary~\ref{cor:necessity of 4 agents scat pbl}. For the sufficient part, we present an algorithm  \textsc{PerpExplore-Scat-Pbl}, that can solve the \textsc{PerpExploration-BBH} problem in the above context. In this algorithm we assumed that 4 agents are initially scattered at 4 different nodes of $R$. To include the remaining cases where the agents are initially scattered at less than or equal to 3 nodes, a slight modification of \textsc{PerpExplore-Scat-Pbl} is enough which we describe in Remark~\ref{remark:modified perpexplore-scat-pbl} in \textbf{Appendix}. Here we describe only the main idea of the algorithm. The detailed description and pseudocodes of the algorithm can also be found at the \textbf{Appendix} (Section~\ref{appendix:pseudo scat pbl}).

 \noindent \textbf{Brief Idea of the Algorithm:} Let us consider, the starting position of $a_0$ to be the first, after which the starting position of $a_1,a_2$ and $a_3$, respectively follows in clockwise order. Let $h_i$ be the starting node of an agent $a_i$ (i.e., $home$ of $a_i$), where $i \in \{0, 1, 2, 3\}$. By $Seg(a_i)$ we define the clockwise arc starting from the node $h_i$ and ending at $h_{(i+1) \pmod 4}$ (called \textit{segment} of $a_i$). If no agents are destroyed by the Byzantine black hole $v_b$ then the exploration of $R$ goes as follows. 
 
 Agent $a_i$ moves clockwise along $Seg(a_i)$ leaving its pebble at $h_i$ until it reaches $h_{(i+1) \pmod 4}$ (i.e., the end of $Seg(a_i)$). An agent can distinguish the node  $h_{(i+1) \pmod 4}$ 
 by seeing the pebble left there by the agent $a_{(i+1)\pmod 4}$. When $a_i$ reaches $h_{(i+1)\pmod 4}$ traversing $Seg(a_i)$ for the first time, it knows the length of the segment $Seg(a_i)$ and stores the length in its own memory. For further traversals it does not depend on seeing pebbles at the end nodes of the segment. It can simply use the length of the segment. After $a_i$ reaches $h_{(i+1)\pmod 4}$, it waits there for a certain number of rounds so that the other agents can in the meantime reach the endpoints of there corresponding segments, while moving along a clockwise direction.  After this waiting, $a_i$ collects the pebble (if exists) from $h_{(i+1)\pmod 4}$ (the pebble left by $a_{(i+1)\pmod 4}$) and starts moving counter-clockwise along with the collected pebble until it reaches its own $home$, i.e., $h_i$ along with the pebble it was carrying. The waiting time at $h_{(i+1)\pmod 4}$ also ensures that all agents starts moving counter-clockwise at the same round. Also note that, if $a_i$ does not find any pebble to collect at $h_{(i+1)\pmod 4}$ it starts moving counter-clockwise without any pebble (this case can only happen if $a_{(i+1) \pmod 4}$ is destroyed at the Byzantine black hole node $v_b$, while it was returning back after collecting the pebble from $h_{(i+2)\pmod 4}$ in some previous round). After $a_i$ reaches $h_i$ moving counter-clockwise, it again waits for another set of rounds before it repeats the whole process again. Again, the waiting time at $h_i$ is configured in such a way that, all agents start repeating the process at the same round (for details on the precise value of these waiting times, refer to the detailed description of the algorithm in \textbf{Appendix} in Section~\ref{appendix:pseudo scat pbl}). Note that, since $\cup_{i=0}^3 Seg(a_i) =R$, if no agents are destroyed, the perpetual exploration of $R$ happens. 
 \begin{figure}
     \centering
     \includegraphics[scale=0.4]{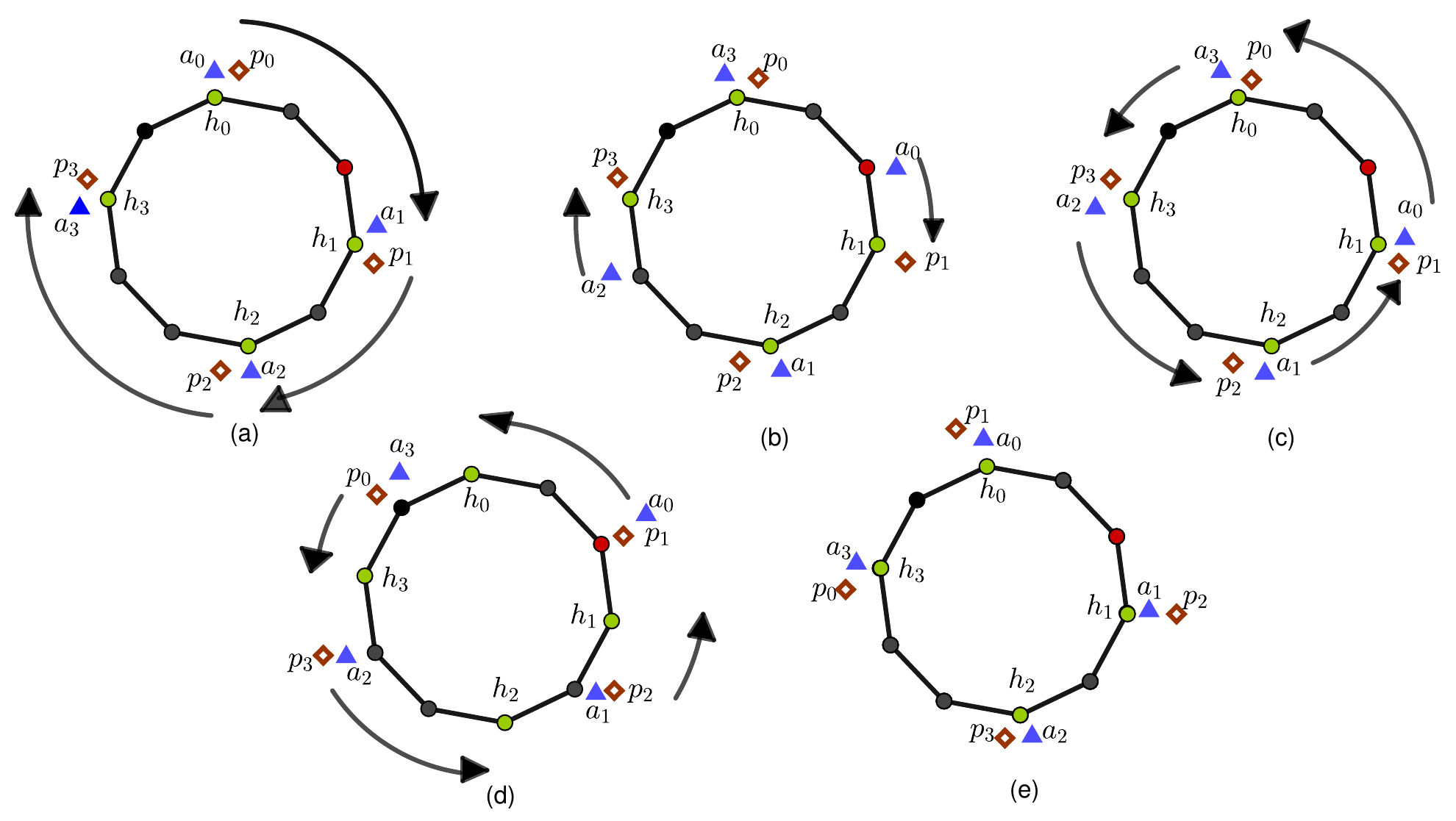}
     \caption{\textbf{(a)}  $h_i$'s are $home$ marked as green. Agent $a_i$ is on $h_i$ initially with a pebble. We name the pebble initially at $h_i$ as $p_i$, but in reality they are anonymous. \textbf{(a-b)} Each agent moves clockwise until the next $home$ (i.e., $h_{(i+1) \pmod{4}}$) without carrying any pebble. The agents already reached ( here $a_1$ and $a_3$) waits for others. \textbf{(c-e)}  All agents are on their clockwise nearest $home$. They start moving counter clockwise together with the pebble present at their current location towards their initial $home$. The agents which already reaches their $home$, wait for others to reach their $home$. } 
     \label{fig:scat pebble}
 \end{figure}
 The exploration can be hampered only if, an agent is destroyed at the Byzantine black hole node $v_b$. Note that, since $Seg(a_i) \cap Seg(a_j)$ is either empty or consists of a safe node and in addition to that, as a segment is explored by only one agent. Hence, at most one agent can be destroyed while exploring their respective segments, as described  above. Without loss of generality, let $v_b\in Seg(a_j)$, for some $j \in \{0,1,2,3\}$. Now there are two cases.\\
 \textit{\underline{Case-I:}} Let $a_j$ is destroyed while moving clockwise. For this case, $a_j$ fails to collect the pebble at $h_{(j+1)\pmod 4}$ left by $a_{(j+1)\pmod 4}$ and return $h_j$. So, when $a_{(j+1)\pmod 4}$, returns $h_{(j+1)\pmod 4}$ after moving counter-clockwise along with the pebble it has collected from $h_{(j+2)\pmod 4}$, it finds two pebble at $h_{(j+1)\pmod 4}$. This is considered by $a_{(j+1) \pmod 4}$ as an anomaly and it learns that $v_b$ must be in the arc, which is in the counter-clockwise direction starting from $h_{(j+1)\pmod 4}$. In this scenario, $a_{(j+1)\pmod 4}$ waits a certain number of rounds (if needed) to ensure that every other alive agents have reached their corresponding $home$.  Then $a_{(j+1)\pmod 4}$ starts moving clockwise with both the pebbles. The aim of this move is to meet and gather with the other alive agents. Note that the other alive agents wait at their corresponding $home$, after moving along the counter-clockwise direction. It is because, at their respective $home$ they do not detect any anomaly. The waiting time is provided in such a way that it is enough for $a_{(j+1)\pmod 4}$ to detect anomaly and meet both the remaining alive agents after moving clockwise, while the other agents are waiting at their $home$. The agent $a_{(j+1) \pmod 4}$ first meets $a_{(j+2)\pmod 4}$ at $h_{(j+2)\pmod 4}$ while it moves clockwise after detecting anomaly. Then both of them moves together, while $a_{(j+2)\pmod4}$ carries the pebble, which it was earlier carrying back to $home$. They move until they meet with $a_{(j+3)\pmod 4}$ at $h_{(j+3)\pmod 4}$. Note that at this moment 3 agents are at a node (which is $h_{(j+3)\pmod 4}$), with at least 3 pebbles (one carried by $a_{(j+3)\pmod 4}$ and two carried by $a_{(j+2)\pmod 4}$). In this case they execute the algorithm \textsc{PerpExplore-Coloc-Pbl} and achieve \textsc{PerpExploration-BBH} of the ring $R$.\\
 \textit{\underline{Case-II:}} Let $a_j$ be destroyed at $v_b$ while it was moving counter-clockwise along with the pebble it collected from $h_{(j+1)\pmod 4}$. Then in this case, when all alive agents return to their corresponding $home$, none of them finds any anomaly as all agent sees exactly one pebble at their $home$. So, they wait and start the exploration again at the same round. Note that in this scenario, all agents except $a_j$, move clockwise to the end points of their corresponding segments, waits there and collects pebble (if any) and moves back to their corresponding $home$ again. Now, since $a_j$ was destroyed earlier, the pebble left by $a_{(j+1) \pmod 4}$ was picked by no agents and thus, while $a_{(j+1)\pmod 4}$ returns back to $h_{(j+1)\pmod 4}$ it finds two pebbles and does the same execution as explained in \textit{Case-I}. 

 The basic idea of this algorithm is to gather three agents with at least 3 pebbles (refer Remark \ref{remark:3 pebble not used}) in the expense of one agent and then execute \textsc{PerpExplore-Coloc-Pbl} to solve the \textsc{PerpExplroation-BBH} problem. \\
 \textbf{Sketch of Correctness:}
 To proof the correctness of the algorithm \textsc{PerpExplore-Scat-Pbl}, we have to prove the following theorem. Here we only provide the proof idea. For details see \textbf{Appendix} (Section~\ref{correct scat pbl}).
 \begin{theorem}
     \label{thm:correct scat pbl}
     Algorithm \textsc{PerpExplore-Scat-Pbl} solves \textsc{PerpExploration-BBH} problem of a ring $R$ with 4 synchronous and  scattered agents under the pebble model of communication where each agents are equipped with a pebble.
 \end{theorem}

First in Corollary~\ref{cor: scat pbl no destruction exploration}, we proved that algorithm \textsc{PerpExplore-Scat-Pbl} guarantees exploration if no agents are destroyed. This corollary is a direct consequence of the Lemma~\ref{lemma:scat pbl no destroy exploration}. Further we showed that if an agent is destroyed then exactly one agent detects anomaly and that agent gathers with the remaining two alive agents within finite rounds  (Lemma~\ref{lemma: gather in an iteration} in \textbf{Appendix}). The rest of the proof follows from Theorem~\ref{thm:colocPblCorrect}.
\vspace{-0.35cm}
\subsection{Whiteboard Model of Communication}
\vspace{-0.1cm}
 The main aim of this section is to prove the following theorem.
 \begin{theorem}
     A team of 3 synchronous agents are necessary and sufficient to solve the problem \textsc{PerpExploration-BBH} on a ring $R$ with $n$ nodes, when each node of $R$ has a whiteboard of $O(\log n)$ bits of memory, irrespective of their starting location. 
 \end{theorem}
 The necessary part follows from Corollary~\ref{cor: necessity 3 agnets whitbrd}. The sufficiency part for co-located initial situation follows from Theorem~\ref{thm:colocwhitbrd}. For the other cases where the agents are initially scattered at more than one starting nodes, we  propose an algorithm \textsc{PerpExplore-Scat-Whitbrd}. This algorithm is designed for the case when all the agents are starting at different nodes. Note that, this algorithm can be easily modified a bit to include the case where initially three agents are scattered at two distinct nodes (refer the modification in \textbf{Appendix}, Section~\ref{Appendix: multiplicity whitbrd case}). For detailed description and the pseudo code, see \textbf{Appendix} (Section~\ref{Appendix: desc scat whitbrd}).
 \\
 \noindent\textbf{Brief Description of the Algorithm:} Let $a_0,a_1$ and $a_2$ be three agents starting from the nodes $h_0,h_1$ and $h_2$, respectively, where these nodes are in a clockwise order. Let $Seg(a_i)$ (also called `Segment of $a_i$') be the clockwise arc starting from $h_i$ and ending at $h_{(i+1)\pmod 3}$. The algorithm first ensures that the ring is explored perpetually if no agents are destroyed. In order to do this, the algorithm goes as follows.

 An agent, say $a_i$, first erases all previously stored data (if any) at $h_i$. Then it writes the message \texttt{(home, ID($a_i$))} at $h_i$ and starts moving clockwise. This type of message is called a \texttt{home} type message that indicates it is $home$ of $a_i$. $a_i$ moves clockwise until it reaches $h_{(i+1)\pmod 3}$. It distinguishes $h_{(i+1)\pmod 3}$ by seeing the \texttt{home} type message left by $a_{(i+1)\pmod 3}$. When $a_i$ moves clockwise, it also marks each node of $Seg(a_i)$, except $h_i$ and $h_{(i+1)\pmod 3}$, by writing \texttt{right}, after erasing any previous such markings (if at all exists) at each such nodes. After $a_i$ reaches $h_{(i+1)\pmod 3}$ it waits for a certain number of rounds (if needed), so that the other agents (say $a_j$) gets enough time to reach endpoints of their corresponding segment (i.e.,$h_{(j+1)\pmod3}$). After this waiting, each alive agent $a_i$ write the message \texttt{(visited, ID($a_i$))} at $h_{(i+1)\pmod 3}$ at the same round. This type of messages is termed as a \texttt{visited} type message, which indicates that an agent with its corresponding ID, also mentioned in the \texttt{visited} type message, has visited the respective node. Then each $a_i$ waits for $n$ rounds at $h_{(i+1)\pmod 3}$ and then starts moving counter-clockwise at the same round. $a_i$ moves counter-clockwise until it reaches $h_i$, i.e., its own $home$. While moving counter-clockwise, $a_i$ erases previously written \texttt{right} marking (which it marked while moving along clockwise direction) from each nodes of $Seg(a_i)$ and writes \texttt{left} there upon arriving (except at $h_i$ and at $h_{(i+1)\pmod 3}$). When $a_i$ reaches its own $home$ (i.e., $h_i$), it waits again for a certain number of rounds, upon seeing the \texttt{visited} type message, left there by $a_{(i-1)\pmod 3}$. This waiting period is enough for each of the other alive agents to reach their corresponding $home$. After this waiting period is over, all of the agents, starts repeating the same procedure again together, from the same round. This procedure ensures that, if no agent is destroyed at the Byzantine black hole node $v_b$, then the perpetual exploration of $R$ continues. This can only be hampered, only if an agent gets destroyed at the node $v_b$. Without loss of generality let $v_b\in Seg(a_j)$ for some, $j \in \{0,1,2\}$. So only $a_j$ can be destroyed at $v_b$, while performing this exploration. It is because $a_j$ is the only agent to visit each node $u$ of $Seg(a_j)$, where $u\in Seg(a_j)\backslash \{h_j, h_{(j+1)\pmod 3}\}$. There can be two cases.\\
 \noindent\textit{\underline{Case-I:}} Let $a_j$ be destroyed while it is moving in clockwise direction along $Seg(a_j)$. This implies it fails to reach $h_{(j+1)\pmod 3}$ and also fails to write a \texttt{visited} type message there. So, when $a_{(j+1)\pmod 3}$ returns to its $home$ (i.e., $h_{(j+1)\pmod 3}$), it finds no \texttt{visited} type message (as it should've, if $a_i$ was not destroyed). $a_{(j+1)\pmod 3}$ interprets from this that, the segment which is adjacent to its own segment in counter-clockwise direction has the Byzantine black hole and an agent must have been destroyed there while moving clockwise. This information triggers it to move clockwise with the aim of gathering with the other agent (i.e., $a_{(j+2)\pmod 3}$).  Note that, when $a_{(j+1)\pmod 3}$ starts moving, the other agent is waiting and the waiting time is sufficient for the moving agent to meet it. After they meet, $a_{(j+1)\pmod3}$ shares the information about its direction of movement in the whiteboard of $h_{(j+2)\pmod3}$. With this, they again start moving clockwise together until they reach $h_j$. The agents can distinguish $h_j$ by the \texttt{home} type message written there by $a_j$ before it was destroyed. Note that in the clockwise direction each node after $h_j$ up to the previous node of $v_b$ were marked \texttt{right} by $a_j$ before it was destroyed. So from $h_j$, $a_{(j+1)\pmod 3}$ and $a_{(j+2)\pmod 3}$ starts moving cautiously clockwise. That is, while both agents ($a_L$ and $a_H$, where $a_L$ and $a_H$ are the agent with lowest ID and highest ID among $a_{(j+1)\pmod 3}$ and $ a_{(j+2)\pmod 3} $ respectively) are at a node $v_0$, $a_L$ moves clockwise to the next node, say $v_1$ while the other agent waits at $v_0$. If at $v_1$, $a_L$ sees the \texttt{right} marking, it interprets $v_1$ is safe. So, it comes back to $v_0$. At $v_0$, seeing that $a_L$ has returned, $a_H$ also interprets $v_1$ to be safe. So, in the next round both $a_L$ and $a_H$ moves to $v_1$ together and repeats the process from $v_1$ again. When $a_L$ reaches $v_b$, it sees no \texttt{right} marking there. $a_L$ if alive, interprets this by determining the current node to be Byzantine black hole. So, it starts perpetual exploration of $R$, except $v_b$. Otherwise, if $a_L$ gets destroyed at $v_b$, it does not return to the previous node where $a_H$ is waiting. Even after waiting, when $a_H$ sees $a_L$ has not returned, it interprets this incident as the next clockwise node is the Byzantine black hole and starts exploring the ring $R$ avoiding that node. Thus for this case the perpetual exploration continues.\\
 \noindent \textit{\underline{Case-II:}} Let $a_j$ be destroyed at $v_b$ while it is moving counter-clockwise. In this case both the alive agents $a_{(j+1)\pmod 3}$ and $a_{(j+2)\pmod 3}$ find \texttt{visited} type message after returning to their corresponding $home$. This is because, $a_j$ is destroyed at $v_b$ after writing \texttt{visited} type message at $h_{(j+1)\pmod 3}$. So, all alive agents after waiting, starts repeating the exploration procedure again. The agents first erase the whiteboard content at their corresponding $home$ and then writes the corresponding \texttt{home} type message before it starts moving clockwise until the end of their corresponding segment. Note that since $a_j$ was destroyed earlier (even before the repetition starts), $a_{(j-1)\pmod 3}$ finds the \texttt{visited} type message which was written there earlier by it, is still present (which was supposed to be erased if $a_j$ was alive). From this information, $a_{(j-1)\pmod 3}$ interprets that, $v_b$ must be in the segment adjacent to its own segment in the clockwise direction. It also interprets that, an agent must have been destroyed at $v_b$ while it was moving counter-clockwise. This is because, if the agent was destroyed while it was moving clockwise then it would have been detected by the other agents when they are at their $home$ (as described in \textit{Case-I}), which is before the start of the next clockwise move (note that the agents have already started next clockwise move). So, this information triggers $a_{(j-1)\pmod 3}$ to move counter-clockwise with the aim to gather with $a_{(j-2)\pmod 3}$. Since  $a_{(j-1)\pmod 3}$ starts its move while $a_{(j-2)\pmod 3}$ waits at $home$ and also, since the waiting time is sufficient, $a_{(j-1)\pmod 3}$ meets  $a_{(j-2)\pmod 3}$ during the waiting period at $h_{(j-1)\pmod 3}$. So, after they meet, $a_{(j-1)\pmod 3}$, communicates the direction of its move to $a_{(j-2)\pmod 3}$ on the whiteboard of $h_{(j-1)\pmod 3}$ and they move together until they reach $h_{(j-2) \pmod 3} (= h_{(j+1)\pmod 3})$ together. They distinguish the node by the \texttt{home} type message left there by $a_{(j-2)\pmod 3}$ before the start of moving clockwise. Note that, in the counter-clockwise direction, the nodes in $Seg(a_j)$, starting from the next node of $h_{(j+1)\pmod 3}$ up to the previous node of $v_b$ are the only nodes those were marked \texttt{left} by $a_j$ before it was destroyed. So, from $h_{(j+1)\pmod 3}$, the alive agents $a_{(j+1) \pmod 3}$ and $a_{(j+2)\pmod 3}$ start moving cautiously in the counter-clockwise direction. Among $a_{(j+1) \pmod 3}$ and $a_{(j+2)\pmod 3}$, an agent with lowest ID is denoted by $a_L$ and the agent with highest ID is denoted by $a_H$. The cautious walk is same as described in Case-I, except for the fact that here after moving one node in the counter-clockwise direction, the agent $a_L$ searches for the \texttt{left} marking. If it finds such marking it returns to $a_H$ and then both moves counter-clockwise and repeats the process. On the other hand if $a_L$ does not find any \texttt{left} marking (it must be $v_b$) and it stays alive, $a_L$ moves out of that node and starts exploring $R$ avoiding that node. On the contrary, if $a_L$ gets destroyed then $a_H$ sees that $a_L$ has not returned while it should have. From this incident it interprets next counter-clockwise node is $v_b$ and it starts exploring $R$ avoiding that node.

 \noindent\textbf{Sketch of Correctness:} Here we give a brief glimpse on how we proved the following theorem. For details see \textbf{Appendix} (Section~\ref{Appendix:correct whitbrd scat})
 \begin{theorem}
 \label{thm: correct whitbrd scat}
     Algorithm \textsc{PerpExplore-Scat-Whitbrd} solves \textsc{PerpExploration-BBH} problem of a ring $R$ with $n$ nodes and with 3 synchronous agents initially scattered under the whiteboard model of communication
 \end{theorem}

 In order to prove this theorem, we first define a \textit{cautious start node}.  Let $a_i$ be the agent destroyed at the Byzantine black hole first. We define \textit{cautious start node} to be $h_i$, if $a_i$ was moving clockwise when destroyed, otherwise it is $h_{(i+1)\pmod 3}$. We first prove that after one agent is destroyed, the remaining two agents gather at the cautious start node. This result can be found in details in
 Lemma~\ref{lemma: reachCautiousStartNode} in \textbf{Appendix}. Note that $a_i$ marks each node on the arc between the cautious start node and the Byzantine black hole with either \texttt{left} or \texttt{right} markings depending on the direction it was moving before it was destroyed. We proved that the alive agent can know the direction  of $a_i$ before it was destroyed. This can be found in details in Lemma~\ref{lemma: agent sees visited type message at other agents home} and Lemma~\ref{lemma: sees no visited at bactrack} in \textbf{Appendix}. After the alive agents reach the cautious start node, they start moving cautiously looking for the marking based on the direction of $a_i$ before it was destroyed. We proved that at least an agent among the two moving cautiously will stay alive knowing the exact location of the Byzantine black hole (Lemma~\ref{lemma: CautiousLeaderCautiousFollower} in \textbf{Appendix}). Hence this agent can explore $R$ perpetually avoiding the Byzantine black hole. This proves Theorem~\ref{thm: correct whitbrd scat}.
 \vspace{-0.5cm}
 \section{Conclusion}\label{section: conclusion}
\vspace{-0.15cm}


The paper addresses perpetual exploration of a ring network in the presence of a malicious node, which we call as a \textit{Byzantine black hole}. This problem, termed as \textsc{PerpExploration-BBH}, is explored under three communication models (\textit{Face-To-Face, Pebble,} and \textit{Whiteboard}) considering various initial scenarios (\textit{co-located} or \textit{scattered} agents) with the aim of minimizing number of agents. We proposed optimal results (in terms of number of agents), for \textit{Pebble} and \textit{Whiteboard} communication models, under both initial scenarios. Further, an upper bound of 5 agents and a lower bound of 3 agents is provided for \textit{Face-to-Face} model, in case of co-located agents.


Future research could focus on proposing an optimal bound for \textit{Face-To-Face} co-located scenario, whereas proposing constructive lower and upper bounds for \textit{Face-To-Face} scattered scenario. Additionally, investigating this problem in different scheduler models can be another way of future direction.


\bibliography{lipics-v2021-sample-article}

\newpage

\begin{center}
    \LARGE \textbf{Appendix}
\end{center}
\section{Proof of Theorem~\ref{impossible:whitebrd2}}
\label{Appendix:thm 3}

  \textbf{Statement:}
   \textit{ A set of two synchronous agents in a ring $R$ of size $n$ cannot solve\\ \noindent \textsc{PerpExploration-BBH}, even in presence of whiteboard if number of possible\\ \noindent consecutive Byzantine black hole positions is at least 3.}\\
\textbf{Proof:}
Let $v_1, v_2 $ and $v_3$ be three possible consecutive Byzantine black hole positions in a ring $R$. Let two agents $a_1$ and $a_2$ be sufficient to solve the \textsc{PerpExploration-BBH} problem on $R$. Thus there exists an algorithm $\mathcal{A}$ such that $a_1$ and $a_2$ can solve \textsc{PerpExploration-BBH} by executing $\mathcal{A}$. Without loss of generality let $a_1$ explores $v_2$ and it moves to $v_2$ for the first time at round $t$ from vertex $v_1$. Then $a_1$ must have been at the vertex $v_1$ at round $t-1$. Let us take two copies $R_1$ and $R_2$  of the same ring $R$.  In $R_1$, $v_1$ is the Byzantine black hole and in $R_2$, $v_2$ is the Byzantine black hole. Let the adversary in $R_1$ destroy $a_1$ at round $t-1$ and in $R_2$ destroy $a_1$ at round $t$. We claim that at round $t$ and $t-1$, $a_2$ can not be on either of $v_1$ or $v_2$. Note that at $t$-th round $a_2$ can not be on $v_2$ otherwise, both $a_1$ and $a_2$ get destroyed at the same round and no other agent is left to explore further. Similarly, at round $t-1$, $a_2$ can not be on $v_1$. Now consider the case when at $t-$ th round $a_2$ is at $v_1$. Since at round $t-1$, $a_2$ can not be on $v_2$  there can be two cases either at round $t-1$, $a_2$ was in $u_0$ ($u_0$ is another adjacent node of $v_1$ except $v_2$) or was in $v_1$ itself. We first argue that at round $t-1$, $a_2$ can not be at $u_0$. Otherwise, since the whiteboard content at nodes except at $v_1$ and $v_2$ are the same in $R_1$ and $R_2$ at round $t$ (due to the same execution by $a_1$ and $a_2$ up to round $t-1$ and at round $t$  for $a_2$ in both $R_1$ and $R_2$). Now since at round $t$, all previous data on the whiteboard except $v_1$ and $v_2$ does not help $a_2$ to distinguish between $R_1$ and $R_2$, the problem now can be thought of solving \textsc{PerpExploration-BBH} with one agent where the number of possible consecutive Byzantine black hole position is atleast two. Now this is impossible to solve. So, at $t-1$-th round $a_2$ can not be on $u_0$.  We now only have to prove that $a_2$ can not be at $v_2$ during round $t-1$. Let in $R_2$, adversary activates the Byzantine black hole at $v_2$ at round $t-1$ which destroys $a_2$  at round $t-1$. Now since the effect of this destruction of agent $a_2$ does not affect the inputs of $a_1$ at round $t-1$ on $v_1$, it moves to $v_2$ at round $t$ where the Byzantine black hole is activated again by adversary destroying $a_1$ too. Thus $a_2$ must remain outside of $v_1$ and $v_2$ during round $t$ and $t-1$ while executing $\mathcal{A}$. So, at round $t$ since all whiteboard content are same for $R_1$ and $R_2$ (except for nodes $v_1$ and $v_2$ whose content only differs after round $t-2$ for $R_1$ and $R_2$ which can not be known by $a_2$ as it was not there during that rounds), they can not help $a_2$ to distinguish between $R_1$ and $R_2$. So the problem now can be thought of as solving \textsc{PerpExploration-BBH} with one agent where the number of possible consecutive Byzantine black hole positions is at least 2. And this is impossible. So there doesn't exist any algorithm that solves \textsc{PerpExploration-BBH} with two agents on a Ring $R$ where the number of possible consecutive positions of the Byzantine black hole is at least 3. 
\section{Proof of Theorem~\ref{impo:pblscat}}
\label{Appendix: thm 6}
\textbf{Statement:} \textit{A set of 3 scattered agents, each equipped with a pebble can not solve the \textsc{PerpExploration-BBH} problem on a ring $R$ with $n$ nodes.}\\
\textbf{Proof:} Let $a_1, a_2$, and $a_3$ be three agents equipped with a pebble each. The agents are placed on three nodes $h_1,h_2$ and $h_3$ initially in such a way that the distance between $h_i$ and $h_j$ is the same for all $ i,j \in \{1,2,3\}$, where we consider $h_j$ to be the nearest node of $h_i$ in the clockwise direction. Without loss of generality let this distance be sufficiently large. Further, let there exist an algorithm  $\mathcal{A}$ that solves the \textsc{PerpExploration-BBH} problem in this setting. Let without loss of generality $a_1$ be the first agent to explore the third node from its corresponding starting position (i.e., $h_1$) in any of the clockwise or counter-clockwise directions, when each agent starts executing the algorithm $\mathcal{A}$. Suppose by following $\mathcal{A}$, $a_1$ can visit the third node in the clockwise direction (without loss of generality) first at a round say $t>0$. Let $v_1, v_2$ and $v_3$ be those sets of three nodes from $h_1$ in the clockwise direction. Let $C_1$,$C_2$ and $C_3$ be three scenarios where in $C_i$, $v_i$ is the Byzantine black hole. We claim that $a_1$ can not carry its pebble during any execution of $\mathcal{A}$. Otherwise, in scenario $C_1$, it would be destroyed along with its pebble, and since the distances between two consecutive $h_i$ are sufficiently large, hence other agents would have no idea that an agent is already destroyed. This is equivalent to solving the \textsc{PerpExploration-BBH} with two agents having a pebble each and as $n$ is sufficiently large this is impossible due to Corollary~\ref{impossible: pebble2}. Now suppose the adversary chooses to activate the Byzantine black hole whenever $a_1$ reaches there for the first time. In this case for all $C_1$, $C_2$ and $C_3$, at round $t$, agents $a_2$ and $a_3$ have no idea about where $a_1$ is destroyed even if they know that $a_1$ is destroyed, as the distances between two consecutive $h_i$ are sufficiently large so their exploration region does not intersect till round $t$ and timeout (or, waiting for other agents) strategy do not work as for all $C_i$ an agent can get same timeout output. Also for all $C_i$s the position of the pebbles and alive agents will be the same at round $t$. Now for the alive agents, the number of nodes for the possible position of the Byzantine black hole must be greater than or equal to 3 at round $t$. So the situation at round $t$ is similar to the problem of solving \textsc{PerpExploration-BBH} on a ring $R$ with two agents having a total of 3 pebbles where the number of possible consecutive positions of the Byzantine black hole is greater or equal to 3. Since we have assumed $\mathcal{A}$ solves the problem thus, $\mathcal{A}$ can also solve the problem of \textsc{PerpExploration-BBH} with two agents where the number of possible positions of the Byzantine black hole is greater or equal to 3. But due to Theorem~\ref{impossible:whitebrd2} it is impossible. Hence there can never exist any algorithm that solves \textsc{PerpExploration-BBH} with three scattered agents each of which is equipped with a pebble.

\section{Algorithm \textsc{PerpExplore-Coloc-Pbl}}
\label{appendix: colocpbl}
\subsection{Detailed Description and Pseudocode}
\label{Appendix:descColoc Pbl}
In the following part, we give a detailed description of our algorithm \textsc{PerpExplore-Coloc-Pbl}. Initially, all the agents are in state \textbf{Initial}. In this state, an agent first declares its $Current-Node$ as $home$, after which initializes the variable $T_{time}=0$ (where, $T_{time}$ is the number of rounds passed since the agent has moved from state \textbf{Initial}) and gathers the ID of the remaining agents currently at $home$. Next, the agent with the lowest ID, i.e., $a_1$ moves to state \textbf{Leader}, the agent with the second lowest ID, i.e., $a_2$ moves to state \textbf{Follower-Find}, whereas the remaining agent, i.e., $a_3$ moves to state \textbf{Backup}. We now define an \textit{iteration} for this algorithm. An iteration is defined to be a collection of $4n+1$ consecutive rounds starting from the latest round where all 3 agents along with 2 pebbles are at $home$ in state \textbf{Initial}. Note that, a new iteration fails to execute if either at least one pebble or an agent gets destroyed by the Byzantine black hole in the current iteration. More precisely, the meaning of failing an iteration implies that an agent after ending its current iteration does not again start a new iteration by moving into state \textbf{Initial}.

Suppose, all the 3 agents along with 2 pebbles successfully execute the $i$-th iteration and reach $home$, i.e., neither any pebble nor any agent gets destroyed by the Byzantine black hole. In this situation, according to our algorithm, $a_1$ being the lowest ID, follows the following sequence $SEQ 1$, whereas $a_2$ being the second lowest ID follows $SEQ 2$, and the $a_3$ follows the sequence $SEQ 3$. The above sequences are as follows:

\begin{enumerate}
    \item $SEQ 1$: \textbf{Initial} $\rightarrow$ \textbf{Leader} $\rightarrow$ \textbf{Initial}.
    \item $SEQ 2$: \textbf{Initial} $\rightarrow$ $({Follower-SEQ})^n$ $\rightarrow$ \textbf{Initial}.
    \begin{enumerate}
        \item ${Follower-SEQ}$: \textbf{Follower-Find} $\rightarrow$ \textbf{Follower-Collect} $\rightarrow$ \textbf{Follower-Find}.
    \end{enumerate}
    \item $SEQ 3$: \textbf{Initial} $\rightarrow$ \textbf{Backup} $\rightarrow$ \textbf{Initial}.
\end{enumerate}

\noindent\textit{\underline{Execution of sequence $SEQ 1$}}: $a_1$ performs this sequence, in which it first changes to state \textbf{Leader}, after finding it as the lowest ID agent at $home$ in state \textbf{Initial}. In state \textbf{Leader}, the agent performs the following checks: if the current node is $home$, then it checks that if the current node has the agent with the second lowest ID, whereas the number of pebbles at $home$ is 2 and $T_{time}\geq 2$. If all these conditions are satisfied, then the agent waits till $T_{time}=4n$ after which it moves to state \textbf{Initial} during $T_{time}=4n+1$. Otherwise, if it finds that it is with the second lowest ID agent and the number of pebbles at $home$ is also 2, but $T_{time}<2$, in this case, the agent initializes $W_{time}=0$ (where $W_{time}$ is the waiting time of the agent) and then moves one hop along with the pebble to the next clockwise node. Otherwise, if the agent finds that neither it is with the second lowest ID agent nor the number of pebbles at $home$ is 2, then it waits for 5 rounds (i.e., waits when $W_{time}< 6$), after which if still the above condition persists, then concludes $S_{lft}$ as the node which is at a clockwise distance of $n-2$ from $home$ and $S_{rgt}$ the node at a clockwise distance of $n-1$. Further, it updates $W_{time}=0$ and moves to the next state \textbf{Detection}. 

On the contrary, if the current node of $a_1$ is not $home$, then also it performs the following checks: first, it checks whether it is with the second lowest ID agent (i.e., $a_2$ in this case), and if so, then initializes $W_{time}=0$ and then moves one hop clockwise to the next node along with the pebble it is accompanying. Otherwise, if it is not with the second lowest ID agent, then wait for 3 rounds (i.e., waits when $W_{time} < 4$), after which if still it is without the second lowest ID agent, then the agent moves to state \textbf{Report-Leader} while leaving the pebble (it is accompanying) at the current node and moving one hop to the clockwise node.

Note that until and unless there are no anomalies detected by the agent $a_1$, it ends an iteration, by changing to state \textbf{Initial} from \textbf{Leader}. Otherwise, if any anomaly is detected, then $a_1$ ends the current iteration by either changing to state \textbf{Detection} or \textbf{Report-Leader} from the state \textbf{Leader}, in which case, it never changes to state either \textbf{Leader} or \textbf{Initial}.

\noindent\textit{\underline{Execution of sequence $SEQ 2$}}: $a_2$ being the second lowest ID agent in state \textit{Initial}, executes this sequence in an iteration. On successful completion of this iteration, an agent starting from \textbf{Initial}, performs the sub-sequence ${Follower-SEQ}$, $n$ times and then again changes to state \textbf{Initial}. The sub-sequence ${Follower-SEQ}$ symbolizes the sequence \textbf{Follower-Find} $\rightarrow$ \textbf{Follower-Collect} $\rightarrow$ \textbf{Follower-Find}, which an agent performs $n$ times until it again changes to state \textbf{Initial} (i.e., ${Follower-SEQ}$ is denoted by ${Follower-SEQ}^n$). While executing ${Follower-SEQ}$, the agent in state \textbf{Follower-Find} checks, whether the current node does not contain the lowest ID agent, and the $Move$ parameter is also set to 0 (the parameter $Move$ is either 0 or 1, when 1 it symbolizes that the current node must contain the lowest ID agent, when 0, it means the current node is the adjacent counter-clockwise node with respect to the node which contains the lowest ID agent), if so then it moves one hop clockwise leaving the pebble at the current node, while updating $Move$ to 1. Otherwise, if it finds $Move$ to be 1 whereas the current node does not contain the lowest ID agent, then it stops the current iteration, detects that the current node is a Byzantine black hole, and starts performing perpetual exploration, avoiding the Byzantine black hole node. On the other hand, if the current node contains the lowest ID agent, then $a_2$ directly changes to state \textbf{Follower-Collect}. In state \textbf{Follower-Collect}, $a_2$ first moves one hop counter-clockwise, then in the current node either finds a pebble or not. If it finds a pebble, then moves one hop clockwise along with the pebble. Now, if $T_{time}=4n+1$ and the current node is $home$, then stops performing $Follower-SEQ$ and moves to state \textbf{Initial}. Otherwise, updates $Move$ to 0 and changes to state \textbf{Follower-Find}. If the pebble is not found, then it concludes the current node to be the Byzantine black hole and starts performing perpetual exploration avoiding this node.

Note that, the agent only changes to state \textbf{Initial} only when it reaches $home$ again, which is after it performs ${Follower-SEQ}$ for $n$ times starting from $SEQ 2$. If at any point, $a_2$ finds any irregularities, based on whether a pebble is present or not, then only it directly concludes the Byzantine black hole position. In all other cases, it follows the ${Follower-SEQ}$ sequence and changes to state \textbf{Follower-Collect}.

\noindent\textit{\underline{Exploration of sequence $SEQ 3$}}: This sequence is only performed by the highest ID agent, which is $a_3$ in this case. In this sequence, after being in state \textbf{Initial}, it changes to state \textbf{Backup}, in which the agent first waits at $home$ till $T_{time}=4n$, after which it checks the following details, and accordingly either moves to state \textbf{Initial} if no irregularities are detected, otherwise moves to state \textbf{Find-Pebble} or \textbf{Find-BH}, based on the irregularity it detects. 
\begin{itemize}
    \item The current node, i.e., $home$ has only the lowest and highest ID agents, i.e., $a_1$ and $a_3$, whereas the number of pebble at $home$ is 0. If this condition is satisfied, then the agent changes its state to \textbf{Find-Pebble}. Note that, this condition is satisfied only when $a_2$ fails to reach $home$ while performing this iteration, i.e., enters the Byzantine black hole, while $a_1$ has detected some anomalies, which instigates it to leave the pebble at the current position at which it detects the anomaly, and henceforth moves back to $home$ while in the state \textbf{Report-Leader}.   
    \item $home$ has only the highest and lowest ID agents, i.e., $a_1$ and $a_3$, but the number of pebbles at $home$ is 1, which implies that $a_2$ has failed to return to $home$ even when $T_{time}=4n$. In this case, $a_3$ concludes $S_{lft}$ to be the node which is at a clockwise distance of $n-2$ from $home$, whereas $S_{rgt}$ to be the node which is at a clockwise distance of $n-1$ from $home$. Finally, it updates $W_{time}=0$ and then moves to state \textbf{Detection}.
    \item $home$ has all three agents and the number of pebbles present at $home$ is 2. In this case, $a_3$ changes to state \textbf{Initial}, completing the iteration without detecting any anomalies.
    \item After $4n$ rounds, there is no other agent except $a_3$ present at $home$. In this case, the agent moves one hop clockwise, and then changes to state \textbf{Find-BH}. This case arises only when both the agents $a_1$ and $a_2$ fail to reach $home$ while performing this iteration.
\end{itemize}

Next we define the states \textbf{Report-Leader}, \textbf{Find-Pebble}, \textbf{Find-BH},  \textbf{Detection}. An agent moves to either of these states only when it detects some anomalies.

The state \textbf{Report-Leader}, is executed by only the lowest ID agent, i.e., $a_1$, while it is in state \textbf{Leader}. The anomaly detected is as follows, $a_1$,  after reaching a new node that is not $home$, along with the pebble it is carrying, waits for at most 3 rounds, for the second lowest ID agent, i.e., $a_2$ to arrive (which is in state \textbf{Follower-Find}). If $a_2$ does not arrive even after the waiting, $a_1$ concludes that the agent (i.e., $a_2$) has entered the Byzantine black hole, this triggers $a_1$ to move to state \textbf{Report-Leader}. In this state, it moves clockwise until it finds the highest ID agent, i.e., $a_3$ at $home$. Then waits until $T_{time}=4n$, after which it changes its state to \textbf{Find-Pebble}.

The state \textbf{Find-Pebble}, is performed together by only $a_1$ and $a_3$, i.e., the lowest and highest ID agent. This state can be reached by $a_1$ only via \textbf{Report-Leader} state, whereas by $a_3$ from the state \textbf{Backup}. The main idea for the agents in this state is to move counter-clockwise till they find the pebble left by $a_1$ when it detects some anomaly (i.e., $a_2$ has entered the Byzantine black hole) and move into state \textbf{Report-Leader}. This pebble acts as a marker which indicates that the Byzantine black hole is either the counter-clockwise neighbor of this node or, it is at a counter-clockwise distance of 2 hops from it. This node at which it finds the pebble is declared to be the new $home$, from which they conclude $S_{lft}$ and $S_{rgt}$ to be the nodes which are at a clockwise distance of $n-2$ and $n-1$, respectively from this new $home$. Finally, they update $W_{time}=0$ and then change to state \textbf{Detection}.
Note that, this state is performed by an agent during the last round of the current iteration, (i.e., when $T_{times}=4n+1$). 

The state \textbf{Find-BH}, is only executed by the highest ID agent, i.e., $a_3$. This state is executed by $a_3$, only when it finds no other agent at $home$, at the last round of the current iteration. This situation can only occur, when both the agents $a_1$ and $a_2$ have entered the Byzantine black hole. So, this triggers $a_3$ to change to state \textbf{Find-BH}, in which it moves along a clockwise direction until a pebble is found. Whenever a pebble is found, it concludes the next node to be the Byzantine black hole, and starts performing perpetual exploration avoiding this node.

Finally, the state \textbf{Detection}, is performed by an agent (i.e., $a_1$ or $a_3$) only if they know the position of $S_{lft}$ and $S_{rgt}$, which are not only two consecutive nodes in a suspicious region $S$. Moreover, these two nodes are exactly at a distance of clockwise $n-1$ and $n-2$ distance from $home$ (the $home$ can be the initial starting node, or it can be the updated $home$ as well). In this state, the agent with the lowest ID, i.e., $a_1$, moves clockwise till it reaches $S_{rgt}$, and then again returns back to $home$. After returning back, if it finds $a_3$, then again it performs the same procedure.

\begin{algorithm}\footnotesize
\caption{\textsc{PerpExplore-Coloc-Pbl}}
\label{algorithm: PerpExplore-Coloc-Pbl}
Input: $n,~k=3$;\\
States:\{\textbf{Initial, Leader, Follower-Find, Follower-Collect, Backup, Report-Leader,Find-Pebble, Detection, Find-BH}\}\\
\underline{In State \textbf{Initial}}:\\
$T_{time}=0$ \tcp{$T_{time}$ is the number of rounds elapsed since the agent has moved from state Initial}
Declare $Current-Node$ as $home$.\\
Gather the IDs of the remaining agents at $home$.\\
\uIf{lowest ID}
{
Move to state \textbf{Leader}.\\
}
\uElseIf{second lowest ID}
{
Set $Move=0$.\\
Move to state \textbf{Follower-Find}.\\
}
\Else
{
Move to state \textbf{Backup}.\\
}
\underline{In State \textbf{Leader}:}\\
\eIf{$Current-Node=home$}
{
\uIf{$ with~Second ~Lowest ~ID \wedge \#Pebble= 2 \wedge T_{time}\geq 2$ }
{
    Wait at the current node till $T_{time}=4n$.\\
    Move to state \textbf{Initial}.\\
    
}
\uElseIf{$ With~Second ~Lowest ~ID \wedge \#Pebble=2 \wedge T_{time}< 2$}
{
$W_{time}=0$.\\
Move one hop along with the pebble to the next node along clockwise direction.\\
}
\ElseIf{$ Not ~with~Second ~Lowest ~ID \lor  \#Pebble\neq 2$}
{
    \eIf{$W_{time} <6$}
    {
        $W_{time}= W_{time}+1$
    }
    {
        $S_{lft}=$ node at a clockwise distance of $n-2$ from $home$\\
        $S_{rgt}=$ node at a clockwise distance of $n-1$ from $home$\\
       Update $W_{time}=0$ then move to state \textbf{Detection}.\\
    }

}

}
{
\eIf{with second lowest ID}
{
$W_{time}=0$.\\
Move one hop along with the pebble to the next node along clockwise direction.\\
}
{
\eIf{$W_{time}<4$}
{
Stay at the current node, $W_{time}=W_{time}+1$.\\
}
{
Move to state \textbf{Report-Leader} while moving one hop in clockwise direction along without pebble.\\
}
}
}
\underline{In State \textbf{Follower-Find}:}\\
\uIf{not with lowest ID $\wedge$ $Move=0$}
{
Move one hop clockwise, without pebble and update $Move=1$.
}
\uElseIf{not with lowest ID $\wedge$ $Move=1$}
{
Detect the $Current-Node$ as Byzantine black hole and continue exploring the ring perpetually avoiding this node.\\
}
\Else
{
Move to state \textbf{Follower-Collect}.\\
}
\underline{In State \textbf{Follower-Collect}:}\\
Move one hop in counter-clockwise direction.\\
\eIf{pebble is found}
{
Collect the pebble, and move one hop clockwise along with the pebble.\\
\If{$T_{time}=4n+1 \wedge Current-Node=home$}
{
Change to state \textbf{Initial}.\\
}
Update $Move=0$ and change to state \textbf{Follower-Find}.\\
}
{
Detect the $Current-Node$ as Byzantine black hole and continue perpetually exploring the ring avoiding this node.
}
\end{algorithm}
\begin{algorithm}\footnotesize
\setcounter{AlgoLine}{53}
    \underline{In State \textbf{Report-Leader}:} \\ 
\If{$T_{time}<4n+1$}
{
Move clockwise until finds the agent with highest ID.\\
Wait until $T_{time}=4n$.\\
}
Change to state \textbf{Find-Pebble}.\\

\underline{In State \textbf{Backup}:}\\
Wait until $T_{time}=4n$.\\
\uIf{$Current-Node$ has only the lowest and highest ID agents and $\#Pebble =0$}
{
Move to state \textbf{Find-Pebble}.\\
}
\uElseIf{$Current-Node$ has only the lowest and highest ID agents and $\#Pebble =1$}
{
 $S_{lft}=$ node at a clockwise distance of $n-2$ from $home$\\
        $S_{rgt}=$ node at a clockwise distance of $n-1$ from $home$\\
       Update $W_{time}=0$ $home\_away=1$ and move to state \textbf{Detection}.\\
}
\uElseIf{$Current-Node$ has both lowest and second lowest ID agent and $\#Pebble=2$}
{
Change to state \textbf{Initial}.\\
}
\ElseIf{$Current-Node$ has no other agent}
{
Move in a clockwise direction and change to state \textbf{Find-BH}.\\
}

\underline{\textbf{In State \textbf{Find-BH}:}}\\

\eIf{a pebble is found}
{
Conclude the next node along clockwise direction is the Byzantine black hole, and continue perpetual exploration avoiding the Byzantine black hole node.\\
}
{
Move along a clockwise direction.\\
}

\underline{In State \textbf{Find-Pebble}:}\\
Move in a counter-clockwise direction.\\
\eIf{a pebble is found}
{
Declare $Current-Node$ as $home$.\\
$S_{lft}=$ node at a clockwise distance of $n-2$ from $home$\\
$S_{rgt}=$ node at a clockwise distance of $n-1$ from $home$\\
       Update $W_{time}=0$ and move to state \textbf{Detection}.\\
}
{
Move along a counter-clockwise direction.
}

\underline{In State \textbf{Detection}:}\\

\eIf{lowest ID}
{
Moves clockwise till the node at a distance $n-2$ from $home$ then returns back to $home$ along the same path and checks for the other agent.\\
If other agent is there it repeats the same procedure as stated in line 86 of Algorithm~\ref{algorithm: PerpExplore-Coloc-Pbl}. Otherwise detects $S_{rgt}$ as the Byzantine black hole and start exploring avoiding that node. 
}
{
Moves one hop counter-clockwise from $home$ then returns back to $home$ along the same path then waits for $2n-6$ rounds and checks for the other agent.\\
If other agent is there it repeats the same procedure as stated in line 90 of Algorithm~\ref{algorithm: PerpExplore-Coloc-Pbl}. Otherwise detects $S_{lft}$ as the Byzantine black hole and start exploring avoiding that node. 
}
 \end{algorithm}
 
On the contrary, if it does not find $a_3$, then it understands that $S_{rgt}$ is the Byzantine black hole node, and hence continues perpetual exploration avoiding this node. 

On the other hand, $a_3$ in this state moves counter-clockwise direction from $home$, i.e., reaches $S_{lft}$, after which it returns back to $home$ again following the same path and then
waits for $2n-6$ rounds at $home$ and then checks for the other agent, i.e., $a_1$.
 
 If it finds $a_1$, then again continues to perform the same movement, otherwise, it concludes that $S_{lft}$ is the Byzantine black hole node, and hence starts the perpetual exploration of the ring just by avoiding this Byzantine black hole node. 

\subsection{Correctness and Complexity}
In this section, we discuss the correctness and complexity of our algorithm \textsc{PerpExplore-Coloc-Pbl}.

\begin{lemma}\label{lemma: either ring is explored or BBH is detected}
    If no agent gets destroyed by the Byzantine black hole, while performing an iteration of Algorithm \ref{algorithm: PerpExplore-Coloc-Pbl}, then our algorithm ensures that in that iteration, either the ring $R$ is explored when no agent knows the exact Byzantine black hole location or, there is one agent that knows exactly the location of the Byzantine black hole.
\end{lemma}

\begin{proof}
    Suppose till the completion of $i$-th iteration $(i>0)$ no agent is destroyed by the Byzantine black hole, then in any iteration $0<j\le i$, the agent $a_1$ follows the sequence $SEQ1$, whereas $a_2$ follows the sequence $SEQ 2$, as instructed starting from the initial $home$ node. Considering this scenario we have two cases:
    \begin{itemize}
        \item \textit{No agents and no pebbles are destroyed:} In this case, while executing the sequence $SEQ 1$, the agent after changing its state to \textbf{Leader} from \textbf{Initial}, visits a new node after every 4 rounds in the clockwise direction. Since no agent has been destroyed in this iteration as well, so $a_1$ un-obstructively visits each node in a clockwise direction, and reaches $home$ to finally end this sequence $SEQ 1$, while again changing its state to \textbf{Initial}. This shows that at least one agent visits each node of $R$ in an iteration completing exploration of $R$.
        \item \textit{No agent gets destroyed but some pebbles are destroyed by the Byzantine black hole:} Note that, this situation can only occur when the Byzantine black hole consumes the pebble carried and released by $a_2$, while $a_2$ is not at the node containing the pebble, executing the sequence $SEQ 2$. It is because, in an iteration, only $a_1$ and $a_2$ moves away from $home$ while executing their respective sequences $SEQ 1$ and $SEQ 2$. And the pebble carried by $a_1$ ( i.e., the agent with the smallest ID) cannot be destroyed without destroying $a_1$, as $a_1$ always carries the pebble at each node while executing its respective sequence. So, in this situation, while $a_2$ ( i.e., the agent with the second lowest ID) in the state \textbf{Follower-Collect} moves back to the adjacent node in the counter-clockwise direction, in order to collect the already released pebble, absence of which triggers $a_2$ to conclude the current node to be the Byzantine black hole, and which leads the agent to immediately leave the Byzantine black hole node. In this case, the consumption of a pebble leads to the Byzantine black hole detection.
    \end{itemize}

\end{proof}

\begin{observation}
    If no agent gets destroyed by the Byzantine black hole while performing an iteration of Algorithm \ref{algorithm: PerpExplore-Coloc-Pbl} but some pebbles are destroyed by the Byzantine black hole, then by Lemma \ref{lemma: either ring is explored or BBH is detected}, one agent (more precisely, the agent with second lowest ID, i.e., $a_2$) detects the exact location of the Byzantine black hole. Further, this agent continues to perpetually explore the ring $R$ by avoiding the Byzantine black hole node. 
\end{observation}

\begin{lemma}\label{lemma: exactly one destoyed}
    If exactly one agent enters the Byzantine black hole while performing an iteration of Algorithm \ref{algorithm: PerpExplore-Coloc-Pbl}, then  within finite additional rounds any of the two following conditions hold:
    \begin{enumerate}
        \item The exact location of the Byzantine black hole is detected by at least one agent.
        \item All alive agents become co-located and they agree about two consecutive nodes $S_{lft}$ and $S_{rgt}$, one among which is the Byzantine black hole.
    \end{enumerate}
\end{lemma}

\begin{proof}
    Note that the agent that has entered the Byzantine black hole is either $a_1$ or $a_2$, based on which we have the following conditions:
    \begin{itemize}
        \item \textit{$a_1$ falls in to the Byzantine black hole}: This situation can only occur when $a_1$ following sequence $SEQ 1$ is in state \textbf{Leader}, visits the Byzantine black hole node along with its pebble. Now, based on the adversarial choice of the Byzantine black hole, the pebble may or may not be destroyed by the Byzantine black hole. Irrespective of which, $a_2$ within 3 additional rounds reaches the Byzantine black hole node while in the state \textbf{Follower-Find}, and the absence of $a_1$ triggers the agent $a_2$ to determine the current node as the Byzantine black hole, and immediately leave the current node.
        \item \textit{$a_2$ falls in to the Byzantine black hole}: This can only happen when $a_2$ is not with $a_1$ at the Byzantine black hole node. Otherwise, both agents will be destroyed, contradicting our claim. Note that, $a_2$ can either be with the pebble or alone, while it is destroyed by the Byzantine black hole. Also note that, since $a_2$ is executing the sequence $SEQ 2$, hence the distance between $a_1$ and $a_2$ can be at most 2. So during an iteration, if $a_2$ is destroyed by a Byzantine black hole, then the Byzantine black hole must be any one of the two consecutive counter-clockwise nodes, from the current position of $a_1$. If the current node is not $home$, in this case, $a_1$ gets to know about this fact, only when it finds the absence of $a_2$ even after waiting for 3 rounds, at the current node. Now, according to our algorithm, it leaves the pebble accompanied by it and moves from the current node in the clockwise direction to $home$ with state \textbf{Report-Leader}. After, reaching $home$, it waits until the end of this iteration, i.e., until $4n$ rounds from the start of this iteration and moves to state \textbf{Find-Pebble}. During the $4n+1$-th round from the start of the current iteration, $a_3$ finds that it is only with $a_1$, whereas no pebble exits, this leads the agent to change its state to \textbf{Find-Pebble}. In this state, both the agents move counter-clockwise from $home$ until they encounter a pebble. Note that, this pebble is the one left by $a_1$ after determining the fact that $a_2$ has entered the Byzantine black hole. So, now both these agents $a_1$ and $a_3$, declare the node with the pebble as $home$, whereas also denote the adjacent counter-clockwise node from the current node as $S_{lft}$, whereas the other node adjacent to $S_{lft}$ in the counter-clockwise direction as $S_{rgt}$. On the other hand, if the current node of $a_1$ is $home$, while $a_2$ is destroyed by the Byzantine black hole, then $a_2$ can never reach $home$ along with the pebble it was carrying by following the sequence $SEQ 2$. So, at the end of this iteration, both $a_1$ and $a_3$ find that there are two agents (i.e., $a_1$ and $a_3$) and one pebble, respectively at $home$. This leads both of them to conclude that, the adjacent counter-clockwise node is $S_{rgt}$, whereas the adjacent counter-clockwise node of $S_{rgt}$ is $S_{lft}$. 
    \end{itemize}
    So, in each case, we showed that either of our two conditions holds.
\end{proof}

\begin{remark}
\label{Remark: exploration with two agents}
    Following from Lemma \ref{lemma: exactly one destoyed}, if exactly one agent gets destroyed by the Byzantine black hole and the exact location of the Byzantine black hole is detected by at least one agent, then our algorithm ensures that the ring is perpetually explored by at least one alive agent while avoiding the Byzantine black hole node. On the contrary, all alive agents become co-located and they agree about two consecutive nodes $S_{lft}$ and $S_{rgt}$, among which one is a Byzantine black hole, then until and unless another agent falls into the Byzantine black hole, our algorithm ensures that the ring is perpetually explored, by the remaining alive agents. Note that, by Lemma \ref{lemma: exactly two agents destroyed} while anyone among them gets destroyed, then the other agent detects the exact Byzantine black hole node, and then further perpetually explores the ring $R$, avoiding the Byzantine black hole node.
\end{remark}

\begin{lemma}\label{lemma: exactly two agents destroyed}
   Algorithm \textsc{PerpExplore-Coloc-Pbl} ensures that if two among three agents get destroyed by the Byzantine black hole, then the remaining agent knows the exact location of the Byzantine black hole within some additional finite rounds without being destroyed.
\end{lemma}

\begin{proof}
    Let us suppose, by contradiction, there exists a round $r$ ($>0$), within which two agents get destroyed by the Byzantine black hole, whereas the alive agent is unable to detect the exact location. Now we have the following cases:

    \begin{itemize}
        \item \textit{Both agents get destroyed by the Byzantine black hole in the same iteration}: Note that in this case, the destroyed agents are $a_1$ and $a_2$, respectively, which while executing the sequences $SEQ 1$ and $SEQ 2$ enters the Byzantine black hole. Also, note that $a_1$ and $a_2$ can only be together at a node when $a_2$ has left the pebble it is carrying in the counter-clockwise node. So, now $a_3$ after the end of this iteration, finds that no agent among $a_1$ and $a_3$ has reached $home$, which triggers $a_3$ to change to state \textbf{Find-BH}. In this state, it moves clockwise, until it finds a pebble (this pebble is the one left by $a_2$ in the counter-clockwise node of the Byzantine black hole). Whenever a pebble is found, $a_3$ concludes that the next node is the Byzantine black hole, which contradicts our claim.
        \item \textit{There exists an iteration where exactly one agent is destroyed and the other one is destroyed during round $r$:} Now, in this case, let $r'< r$ be a round when exactly an agent is destroyed during an iteration, say $i >0$. Then by Lemma~\ref{lemma: exactly one destoyed}, within finitely many additional rounds $r_0$, where $r_0+r' < r$, both the alive agents become co-located and agree about two consecutive nodes $S_{lft}$ and $S_{rgt}$, one among which must be the Byzantine black hole. Observe that according to our algorithm, the node that is clockwise adjacent to $S_{rgt}$ is determined to be the $home$ (if not already so). Note that in the other case according to Lemma ~\ref{lemma: exactly one destoyed} where after exactly one agent is destroyed by the Byzantine black hole during an iteration at a round $r'<r$, there is an agent that knows the exact location of the Byzantine black hole within a finite additional round and the agent starts exploring the ring avoiding the Byzantine black hole and so even after round $r$ it stays alive knowing the exact location of the Byzantine black hole contradicting our assumption. Now, when both the agents are co-located during a round say $r_0+r'<r$ and know that the Byzantine black hole must be any of the two consecutive nodes $S_{lft}$ or $S_{rgt}$. In this case, they move into state \textbf{Detection}. After which one agent moves clockwise until it reaches $S_{lft}$ and the other one moves counter-clockwise until it reaches $S_{rgt}$ and moves back to $home$. The agent that visits $S_{rgt}$ waits for the other agent at $home$ for $2n-6$ rounds which is sufficient for the other agent to return back at $home$. Now when during round $r$, another agent among these two gets destroyed, within at most $(2n-4)$ rounds (among the total $2n-6$ rounds, within which both agents are to meet) they fail to meet at $home$. Which triggers the remaining alive agent to determine exactly which one among $S_{lft}$ or $S_{rgt}$ is the Byzantine black hole. This also contradicts the claim we made earlier. Thus, the algorithm \textsc{PerpExplore-Coloc-Pbl} ensures that, if two among three agents get destroyed by the Byzantine black hole, the remaining agent must know the exact location of the Byzantine black hole in some finite additional rounds without getting destroyed.
    \end{itemize}
\end{proof}
\label{Appedix: coloc pbl correct}

\section{Algorithm \textsc{PerpExplore-Scat-Pbl}}
\label{Appendix:perpexplore scat pbl}
\subsection{Detailed description and Pseudocode of the Algorithm}
\label{appendix:pseudo scat pbl}
In this section, we discuss the model where the agents are placed arbitrarily along the nodes of the ring $R$ (note that each such node must be a safe node), where each agent has a movable token, which it can carry along with it, and acts as a mode of inter agent communication. Moreover, the agent can gather the IDs of other agents which are currently at the same node at the same round. With this context, in the following part we show that 4 scattered agents with a pebble each is sufficient to solve \textsc{PerpExploration-BBH} on $R$, using our algorithm \textsc{PerpExplore-Scat-Pbl}.

If 4 agents are scattered in more than one nodes initially then there are  3 cases. Firstly, all four agents are at different nodes. Secondly, four agents are scattered in three different nodes nodes initially. For this case there exists exactly one node with two agents and remaining two nodes with exactly one agents each. In the third case, four agents are scattered in two nodes initially. In this case either each of the two nodes contains exactly two agents or, one node contains three agents and the other one has exactly one agent. We here describe and present the pseudocode considering the first case only. The case with 3 agents in one node can be dealt by instructing the co-located agents to execute \textsc{PerpExplore-Coloc-Pbl} when there are three agents on the current node (i.e., from the beginning). The algorithm we discuss here can also be used for the remaining cases with slight modifications. These modifications are described in the Remark~\ref{remark:modified perpexplore-scat-pbl}

The main idea of the \textsc{PerpExplore-Scat-Pbl} algorithm is that 4 agents perpetually explore the ring, when no agents are destroyed by the Byzantine black hole. Our algorithm ensures that during this exploration at most one agent can be destroyed. In this scenario, the remaining agents within further finite time, gather at a single node (which is the starting node of either of these 4 agents, hence it is safe), and after which they together start executing the algorithm \textsc{PerpExplore-Coloc-Pbl}. 

\begin{remark}\label{remark:3 pebble not used}
    As mentioned in \textsc{PerpExplore-Coloc-Pbl} requires, 3 agents and 2 pebbles to perform \textsc{PerpExploration-BBH}, but in this case we instruct the agents to perform \textsc{PerpExplore-Coloc-Pbl}, whenever three agents and at least three pebbles are able to gather at a node. These agents performs \textsc{PerpExplore-Coloc-Pbl} with a slight modification that the excess pebble (or pebbles) remain at their gathered node, which they term as $home$. So, all the conditions that are mentioned, remains same in this case as well, only the condition where an agent changes to certain state $s_i$ by observing $t_i$ (say) number of pebbles at $home$, transforms to the fact that the agent changes to state $s_i$ by observing $t_i+x$ (where $x=\text{total \# pebbles}-2$) number of pebbles at $home$. 
\end{remark}

At the beginning of the Algorithm \ref{algorithm: PerpExplore-Scat-Pbl} all agents are in state \textbf{Initial1}. In this state, the agents declare their current node as $home$, moreover they initialize the variables $size=0$ and $s=0$, where $size$ variable is used by the agents in order to store the length of the path between it's own $home$ and the nearest $home$ of  another agent in the clockwise direction. Further the variable $s$ indicates the state transition. If the agent is at state \textbf{Forward} and $s=$ 0 then the agent moved from state \textbf{Initial1} to its current state (i.e., state \textbf{Forward}), otherwise if the agent is at state \textbf{Forward} and $s=1$ then this indicates it changes to state \textbf{Forward} from state \textbf{Wait2}. An agent $a_1$ in state \textbf{Forward}, initializes $W_{time}=0$ (which stores the waiting time) and checks whether $s=0$ or $s=1$. Note that, if $s=0$ and the current node is $home$, then the $a_1$ is first instructed to leave the pebble at current node and then move clockwise direction, after increasing $size$ by 1. The next node where $a_1$ has reached can either be $home$ of another agent, say $a_2$, or it is not $home$ for any other agents. For the first case $a_1$ would find a pebble on its $Current-Node$ which was left by $a_2$. In this case $a_1$ updates $s$ to 1 and moves into state \textbf{Wait1}. On the other hand, if the $Current-Node$ is not $home$ of any other agents then $a_1$ finds no pebble at the $Current-Node$ and thus moves clockwise after increasing the variable $size$ by 1.
Note that, when $s=0$ and an agent is in state \textbf{Forward} it always finds another pebble on some node, if it is not destroyed (this pebble must be left by another agent which is simultaneously executing Algorithm \ref{algorithm: PerpExplore-Scat-Pbl} and currently in state \textbf{Forward}). This node must be the clockwise nearest $home$ (i.e., $home$ of another agent which is clockwise nearest to own $home$). Now since during each move at state \textbf{Forward}, an agent increases the variable $size$ by 1, the $size$ variable must store the length between an agents own $home$ and the clockwise nearest $home$ when the agent finds another pebble. Beyond this point (i.e., when $s=1$) the agents will use the $size$ variable as reference to calculate how much it should move, as after this, there is a possibility that the agent may not find a pebble at the clockwise nearest $home$. 

On the contrary, if $s=1$, this implies the agent already knows how much distance it should travel from its own $home$ in order to reach the clockwise nearest $home$. So, now if the current node is $home$ it just moves clockwise by initializing the variable $distance$ to 1 (where the variable $distance$ stores the amount of distance the agent has traversed from its own $home$). Otherwise, if current node is not $home$, then the agent continues to move clockwise while updating $distance$ by 1, until it reaches the clockwise nearest $home$ (i.e., when $distance=size$), after which it further changes to state \textbf{Wait1}. 

In state \textbf{Wait1}, the agent is instructed to wait at the current node (more specifically, the current node is the clockwise nearest $home$ of the agent) for $n-size-1$ rounds. After which it changes to state \textbf{Fetch}. Note that, if each agent starts executing the state \textbf{Forward} at the same round, this implies that these agents change to state \textbf{Fetch} together at the same round (which is precisely at the $n+1$-th round since the start of state \textbf{Forward}).
Thus we have the following observation.
\begin{observation}
    \label{obs: start of fetch} If all agents moved into the \textbf{Forward} state together at a round $t$ then, all of them move to state \textbf{Fetch} together at the round $t+n+1$.
\end{observation}

In state \textbf{Fetch}, the agent first checks if the current node has a pebble, if so then it further checks whether current node is $home$ or not. If the current node is not $home$, then it moves counter-clockwise with pebble, until it reaches $home$. Otherwise, if current node is $home$ then it checks whether there are more than one pebble present at the current node or not. If so, then the agent changes its state to \textbf{Gather1}. Otherwise, if the current node does not have more than one pebble, then it initializes $W_{time}$ to 0 and changes to state \textbf{Wait2}.

In state \textbf{Wait2}, the agent first checks whether $W_{time}<2n-size$, if so then it further checks whether $\#Agent$ at current node is exactly 1. If these two conditions are satisfied then the agent just increments $W_{time}$ by 1. On the other hand, if it finds that $W_{time}<2n-size$ but $\#Agent$ at the current node is 2, then it changes its state to \textbf{Gather2} and moves clockwise and updates $W_{time}$ to 0. Otherwise, if it finds $W_{time}<2n-size$ but $\#Agent$ at the current node is 3 then it updates $W_{time}$ to 0, whereas changes its state to \textbf{Coloc}. Lastly, whenever it finds $W_{time}=2n-size$, it changes its state to \textbf{Forward}.

Note that if no agent gets destroyed by the Byzantine black hole, then our algorithm guarantees that an agent can neither changes its state to \textbf{Gather1} or in state \textbf{Gather2}. It is because, an agent, say $a_1$ can only find more than one pebble at its $home$ node, only when another agent say $a_2$, which was supposed to \textit{fetch} the pebble left by $a_1$ at its $home$, is destroyed. In this case, whenever $a_1$ returns while in state \textbf{fetch} with another pebble, it finds more than one pebble. So the first agent, which first finds this anomaly changes its state to \textbf{Gather1} from the state \textbf{Fetch}. Now this agent must initiate gathering with the remaining two agents (as it is the one which has first detected the \textit{anomaly} of more than one pebble at its $home$). Note that, while $a_1$ starts its gathering phase, moving in a clockwise direction, there is a possibility that other agents are still moving either in clockwise or counter-clockwise direction. So, to address this challenge, the algorithm instructs $a_1$ (i.e., the agent which is in state \textbf{Gather1}) to wait at the current node (i.e., its $home$) for precisely $n-size-1$ rounds, and then start moving in clockwise direction. An agent moves to state \textbf{Gather1} at $(n+2+size$)-th round from the round when all agents changed state to \textbf{Forward} ($n+1$ round to move into \textbf{Fetch} state, $size$ rounds to reach $home$ and one round to change state to \textbf{Gather1}). Then waits for another $n-size-1$ before it starts moving. So it starts moving at $(2n+2)-$th round  from the round when all agents changed state to \textbf{Forward} together the last time. Now, an agent reaches its $home$ in state \textbf{Fetch} in round $n+1+size (\le 2n-1 <2n+2$,\ as $size \le n-2$) after the round when all agents changed to state \textbf{Forward} together the last time. So, the waiting time of $a_1$, being in state \textbf{Gather1} guarantees that all the other alive agents reach their respective $home$ while in state \textbf{Fetch}. Now, in order to avoid further movement of remaining agents (since these agents are yet to detect any anomaly, they might change to state \textbf{Forward} and starts moving, which makes gathering a bit challenging), our algorithm instructs the remaining agents to wait at the current node for a certain number of rounds (i.e., more precisely, $2n-size-1$). This waiting period is sufficient, for the agent $a_1$ to meet with the other alive agents while in state \textbf{Gather1}. Note that, if no agent detects any anomaly, then they directly change their state to \textbf{Wait2} from the state \textbf{Fetch}. If no such agent exists which encounters any anomaly and changes its state to \textbf{Gather1} or \textbf{Gather2} then this waiting time of $2n-size-1$ rounds, also ensures that all of them must move to state \textbf{Forward} at the exact same round (i.e, ($3n+2$)-th round from the previous time it moved to state \textbf{Forward})\footnote{At $(n+1)$-th round an agent moves to state \textbf{Fetch} after that it moves for $size$ rounds to reach its $home$ in counter-clockwise direction and changes to state \textbf{Wait2}. So at round $n+3+size$ an agent changes to state \textbf{Wait2}. After that waits for further $2n-size-1$ rounds and changes to state forward, taking a total of $3n+2$ rounds }. Now we have the following observation
\begin{observation}
\label{obs: move to state forward again}
    Let $t$ be a round when all alive agents moved into state \textbf{Forward}. If no agent detects any anomaly (i.e., no agent finds more than one pebble at home) then, at round $t+3n+2$ all the agents together move into state \textbf{Forward} again.
\end{observation}
\begin{algorithm}[]\footnotesize
\caption{\textsc{PerpExplore-Scat-Pbl}}
\label{algorithm: PerpExplore-Scat-Pbl}
    Input: $n$\\
    State:\{\textbf{Initial1, Forward, Fetch, Wait1, Wait2, Coloc, Gather}\}\\
    \underline{In State \textbf{Initial1:}}\\
    Declare $Current-Node$ as $home$.\\
    Initialize $size=0$ and $s=0$, then change to state \textbf{Forward}.\\
    \underline{In State \textbf{Forward:}}\\
    Initialize $W_{time}=0$.\\
    \eIf{$s=0$}
    {

    \uIf{$Current -Node$ is not $home$ and has a pebble}
    {
       update $s=1$ and
        Change state to \textbf{Wait1}\\
    }
    \uElseIf{$Current -Node$ is $home$}{
        Leave the pebble at $home$.\\
    Move in a clockwise direction and update $size=size+1$.\\
    }
   \Else{
     Move in a clockwise direction and update $size=size+1$.\\
    }
    
    }
    {
    \uIf{$Current-Node$ is $home$}
    {
       Initialize $distance=1$ and move in a clockwise direction leaving the pebble.\\
    }
    \uElseIf{$distance< size$}
    {
    Move in a clockwise direction and update $distance = distance +1$.\\
    }
    \ElseIf{$distance=size$}
    {
     Change State to \textbf{Wait1}.\\
    }
    }
    \underline{In State \textbf{Wait1}:}\\
    \eIf{$W_{time}<n-size$}
    {
    $W_{time}=W_{time}+1$.\\
    }
    {
     Change state to \textbf{Fetch}.\\
    }
    \underline{In State \textbf{Fetch}:}\\
    \eIf{$Current-Node$ has a pebble}
    {
    \eIf{$Current-Node$ is not $home$}
    {
    Move counter-clockwise with the pebble.\\
    }
    {
    \eIf{$home$ has more than one pebble}
    {
    Change to state \textbf{Gather1}.\\
    }
    {
    Initialize $W_{time}=0$ and change to state \textbf{Wait2}.\\
    }
    }
    }
    {
    \eIf{$Current-Node$ is not $home$}
    {
    Move in a counter-clockwise direction.\\
    }
    {
    Initialize $W_{time}=0$ and change to state \textbf{Wait2}.\\
    }
    }
    \underline{In State \textbf{Wait2}:}\\
    \uIf{$W_{time}<2n-size \wedge \#Agent$ at $Current-Node$ is 1}
    {
    $W_{time}=W_{time}+1$.\\
    }
    \uElseIf{$W_{time}<2n-size \wedge \#Agent$ at $Current-Node$ is 2}
    {
    Initialize $W_{time}=0$ and change to state \textbf{Gather2} and moves clockwise.\\
    }
    \uElseIf{$W_{time}<2n-size \wedge \#Agent$ at $Current-Node$ is 3}
    {
    Initialize $W_{time}=0$ and change to state \textbf{Coloc}.\\
    }
    \ElseIf{$W_{time}=2n-size$}
    {
    Change to state \textbf{Forward}.\\
    }
    \underline{In State \textbf{Gather1}:}\\
    \eIf{$W_{time}<n-size$}
    {
    $W_{time}=W_{time}+1$ 
    }
    {
    \eIf{$\#Agent$ at $Current-Node <$3}
    {
    Move in a clockwise direction with all the pebbles at the $Current-Node$\\
    }
    {
    Change to state \textbf{Coloc}.
    }
    }
    \underline{In State \textbf{Gather2}:}\\
    \eIf{$\#Agent$ at $Current-Node$ < 3}
    {
    Move in a clockwise direction with all the pebbles at the $Current-Node$.\\
    }
    {
    Change to state \textbf{Coloc}.\\
    }
    \underline{In State \textbf{Coloc}:}\\
    Declare $Current-Node$ as $home$.\\
    Change state to \textbf{Initial} and execute \textsc{PerpExplore-Coloc-Pbl}.\tcp{State \textbf{Initial} is the state defined in \textsc{PerpExplore-Coloc-Pbl} algorithm}
    
\end{algorithm}
An agent in state \textbf{Gather1}, is the first agent to detect the anomaly, hence it is first instructed to wait at the current node, i.e., its $home$ for the first $n-size-1$ rounds, and then start moving in a clockwise direction, while accompanying all the pebbles at the current node \footnote{Note that we have considered that in this case an agent, while in state \textbf{Gather1} or \textbf{Gather2} can carry more than one pebble, which can be easily restricted to one agent can carry only one pebble, but in that case the agent needs to return and carry one pebble each time and reach back to its current node}. 

This movement in the clockwise direction continues until the number of agent at the current node is exactly 3, after which it changes to state \textbf{Coloc}.

An agent in state \textbf{Gather2} (i.e., the alive agents which are not the first to detect the anomaly), first checks whether the current node has exactly 3 agents on it, if not then it moves in a counter-clockwise direction, while accompanying the pebbles present at its $home$, until it satisfies the condition, where it finds a node on which a total of 3 agents exists. After which they change their state to \textbf{Coloc}.

In state \textbf{Coloc}, an agent declares its current node as $home$, while also changes its state to \textbf{Initial} \footnote{The state \textbf{Initial} is defined in Algorithm \ref{algorithm: PerpExplore-Coloc-Pbl}} and then executes algorithm \textsc{PerpExplore-Coloc-Pbl}.

\subsection{Correctness and Complexity}
\label{correct scat pbl}
In this section, we discuss the correctness and complexity of our algorithm \textsc{PerpExplore-Scat-Pbl}.
Before that, let us define an \textit{iteration} of the algorithm \textsc{PerpExplore-Scat-Pbl}. An iteration is a set of $3n+1$ consecutive rounds starting from the round in which all alive agents change to state \textbf{Forward} together.
Next we prove the following lemma.

\begin{lemma}
\label{lemma:scat pbl no destroy exploration}
    If no agents are destroyed by the Byzantine black hole until the completion of $p$-th iteration, then for each $j-$th iteration, every node of the ring $R$ is visited by at least one agent, where $j\in \mathbb{N}^+,~\text{and}~ j\le p$.
\end{lemma}
\begin{proof}
   Let $a_0,a_1,a_2$ and $a_3$ be 4 agents initially located at $h_0,h_1,h_2$ and $h_3$ respectively where, $h_0,h_1,h_2,h_3$ are in clockwise order starting from $h_0$.  By $Seg(a_i)$ we denote the clockwise arc starting from $h_i$ and ending at $h_{(i+1)\pmod 4}$.
    Note that in any of the $j-$th iteration ($j \le p$), an agent $a_i$ (for all $i \in \{0,1,2,3\}$) moves clockwise
    along $Seg(a_i)$ starting from $h_i$ while in state \textbf{Forward} and reaches $h_{(i+1)\pmod 4}$. Then $a_i$ in the same iteration moves back to $h_i$ again while being in state \textbf{Fetch}. Thus each node of $Seg(a_i)$ is explored twice in $j$-th iteration by $a_i$, where $j \le p$.
    Since for any $v \in R$, $v \in Seg(a_j),$ for some $j \in \{0,1,2,3\}$, and in each iteration $a_j$ explores each node of $Seg(a_j)$ until $p-$th iteration, the result follows.
\end{proof}
\begin{corollary}
\label{cor: scat pbl no destruction exploration}
   If no agents are destroyed, the agents can perpetually explore a ring execuing the algorithm \textsc{PerpExplore-Scat-Pbl}.
\end{corollary}

Next we prove, after the consumption of one agent by the Byzantine black hole, during a subsequent iteration or at the current iteration, at most one agent can change its state to \textbf{Gather1} and gathers with two more agents before that iteration ends.

\begin{lemma}
\label{lemma: gather in an iteration}
In an iteration, at most one agent can change its state to \textbf{Gather1} executing \textsc{PerpExplore-Scat-Pbl}. Furthermore, before the end of that iteration the agent gathers with two more agents at a safe node with at least 3 pebbles.
\end{lemma}
\begin{proof}
   Let $a_0,a_1,a_2$ and $a_3$ be 4 agents initially located at $h_0,h_1,h_2$ and $h_3$ respectively where, $h_0,h_1,h_2,h_3$ are in clockwise order starting from $h_0$. By $Seg(a_i)$ we denote the clockwise arc starting from $h_i$ and ending at $h_{(i+1)\pmod 4}$. Without loss of generality, let the first agent that is destroyed by the Byzantine black hole be $a_0$ during the $p-$th iteration. So, the Byzantine black hole $v_b$, must be in $Seg(a_0)$. Since, no agents changes to state \textbf{Gather1}, before an agent is destroyed. Hence, until $(p-1)-$th iteration, there does not exist any agent, which changes its state to \textbf{Gather1}. Now, we have following two cases, which are based on the state $a_0$ was in, before it gets destroyed by the Byzantine black hole.

     \textit{\underline{Case-I:}} $a_0$ was in state \textbf{Forward} when it was destroyed. So, $a_0$ fails to reach $h_1$ to bring the pebble, left by $a_1,$ back to $h_0$ during $p-$th iteration. So, when $a_1$ reaches $h_1$ at the $p-$th iteration, along with the pebble it collected from $h_2$, it finds two pebble at the node $h_1$. This happens at $(n+2+|Seg(a_1)|)-$th round of $p-$th iteration. In the same round $a_1$ changes to state \textbf{Gather1}. Next it waits an additional $n-|Seg(a_1)|-1$ rounds and then in the $(2n+2)-$ th round of the $p-$ th iteration, it starts moving clockwise along with both the pebbles. Note that, all other alive agents (i.e., $a_2$ and $a_3$) reaches their home by the $(2n+2)-$th round of the $p-$th iteration. Also upon reaching their respective $home$, $a_2$ and $a_3$ both sees exactly one pebble (i.e., the pebble they carried back to their $home$) so they change their state to \textbf{Wait2}. By Observation~\ref{obs: move to state forward again}, from $(2n+2)$-th round of the iteration upto the end of the iteration (i.e., $(3n+1)-$th round of the iteration), both $a_2$ and $a_3$ stays at $h_2$ and $h_3$, respectively in state \textbf{Wait2}. So, after $a_1$ starts moving clockwise in the  $(2n+2)-$th round, it first meets $a_2$ in round $2n+1+|Seg(a_1)|$ (which is $\le 3n-1$ as $|Seg(a_1)|\le n-2$) of the $p-$th iteration. And then meets $a_3$ at round $2n+1+|Seg(a_1)|+|Seg(a_2)| \le 3n$ (as $|Seg(a_1)|+|Seg(a_2)|\le n-1)$ which is also in $p-$th iteration. Now note that, the moment $a_1$ meets $a_2$, it changes to state \textbf{Gather2} and starts moving together with $a_1$ along with its own pebble until they meet $a_3$. When all three of them gathers together at $h_3$ there is atleast 3 pebbles at $h_3$ (two pebbles carried by $a_1$, one pebble carried by $a_2$, and one pebble remains for $a_3$). In this scenario, they change their state to \textbf{Coloc}. Note that in $p-$th iteration only $a_1$ changes to state \textbf{Gather1} as other agents can not change to state \textbf{Gather1} from either from \textbf{Wait2}, or, \textbf{Gather2}, or, \textbf{Coloc}. So, for this case the claimed result holds.

\textit{\underline{Case-II:}} Let $a_0$ was in state \textbf{Fetch} while it was destroyed by the Byzantine black hole node $v_b$. In this case, none of the alive agents see more than one pebble at their $home$ after they return there in the $p-$th iteration. This is because, here the the pebbles left by agents  $a_0, a_1, a_2$ and $a_3$ has been collected by agents $a_3, a_0, a_1 $ and $a_2$, respectively and when they reach their corresponding $home$, the pebble they carried are the only one there. Thus all of them (except $a_0$ as it was destroyed during $p-$th iteration in state \textbf{Fetch}) moves to state \textbf{Wait2} after reaching $home$ at the $p-$th iteration and after the end of the iteration moved to state \textbf{Forward} and starts the $(p+1)-$th iteration. Since $a_0$ was destroyed in the previous iteration it remains unavailable to collect the pebble from $h_1$, left by $a_1$ in the $(p+1)-th$ iteration. Note that here in this iteration both $a_1$ and $a_2$ collects pebble from $h_2$ and $h_3$ ( left by $a_2$ and $a_3$ respectively) and returns back to $h_1$ and $h_2$ respectively. But $a_3$ upon reaching $h_0$ does not find any pebble to collect. In this case it returns to $h_3$ without any pebble. Note that here only $a_1$ sees more than one pebble upon returning back to $h_1$ in the $(p+1)-$th iteration. So only $a_1$ changes to state \textbf{Gather1} at round $n+2+|Seg(a_1)|$. Now by similar argument as in Case-I, the rest follows.
\end{proof}

After 3 agents gather at a safe node with at least 3 pebbles they execute the algorithm \textsc{PerpExplore-Coloc-Pbl} which solves the problem of \textsc{PerpExploration-BBH} on a ring $R$ for 3 co-located agents with at least one pebble each agent (refer, Theorem~\ref{thm:colocPblCorrect}).
Thus this proves Theorem~\ref{thm:correct scat pbl}. which is stated again below.

\noindent\textbf{Statement of  Theorem~\ref{thm:correct scat pbl}:} \textit{ Algorithm \textsc{PerpExplore-Scat-Pbl} solves \textsc{PerpExploration-BBH} problem of a ring $R$ with 4 synchronous and  scattered agents under the pebble model of communication where each agents are equipped with a pebble.}

\subsection{Modification of \textsc{PerpExplore-Scat-Pbl} to include the cases where starting positions has multiplicity}
\label{subsec: include all cases scat pbl}
 \begin{remark}
\label{remark:modified perpexplore-scat-pbl}
   In the discussion of our Algorithm \ref{algorithm: PerpExplore-Scat-Pbl}, we have considered that 4 agents are scattered along 4 distinct nodes (i.e., each node with multiplicity 1). Here we describe how our algorithm (Algorithm~\ref{algorithm: PerpExplore-Scat-Pbl}) works for the remaining cases, (i.e., when 4 agents are scattered among 3 or 2 nodes initially) by slight modification. Observe that, in these remaining cases there exists at least one node, where multiplicity is greater than 1 initially. Let there exists a node with multiplicity $3$, in this case, these agents directly start executing \textsc{PerpExplore-Coloc-Pbl}. While execution of the algorithm, if they encounter the fourth agent somewhere along $R$, they just ignore this agent (note that the IDs, of all the 3 initially co-located agents are already collected by each other) while executing their current algorithm. On the other hand, if initially there exists a node with multiplicity greater than 1 and less than 3, then in that case, only the lowest ID agent (say $a_1$) at the current node, changes its state \textbf{Forward}, whereas the other agent at the current node (say $a_2$) changes its state \textbf{Backup-Wait} at some round say $t_0$ $(>0)$. The agent in state \textbf{Forward}, continue executing the algorithm \textsc{PerpExplore-Scat-Pbl}, whereas the agent in state \textbf{Backup-Wait}, waits at the current node for $2n+1$ rounds (i.e., until $t_0+2n+1$-th round) and then checks for anomaly in the next round (i.e., if the current node has more pebble than current number of agents). If such anomaly exists after $t_0+2n+1$ round, then that implies that $a_1$ is already in state \textbf{Gather1}. It is because, $a_1$ while it returned back to its $home$ carrying a pebble in state \textbf{Fetch} at round $t_0+n+size+2$ round (where $size$ is the length of $Seg(a_1)$),  it encounters the first anomaly, i.e., the number of pebble is more than number agents. In that case, $a_1$ directly changes its state to \textbf{Gather1}, and further waits at $home$ for $n-size-1$ rounds, i.e., till $t_0+2n+1$ rounds. After which it starts to move clockwise, and on the other hand $a_2$ also starts moving clockwise, while changing its state to \textbf{Gather2} from \textbf{Backup-Wait}. This guarantees that if the anomaly is detected at a multiplicity then both $a_1$ and $a_2$ moves together to gather with the third agent, which is in state \textbf{Wait2}. 
   Now if the anomaly is not detected at the multiplicity, that is $a_2$ at round $t_0+2n+1$ finds no anomaly then it moves into state \textbf{Backup} whereas $a_1$ is in state \textbf{Wait2} at the same node. Both of them waits until round $3n+1$ and at round $3n+2$, $a_1$ changes its state to \textbf{Forward} again (refer to Observation~\ref{obs: move to state forward again}) and $a_2$ changes its state to \textbf{Backup-Wait} again. If some other agent detects anomaly it must meet $a_1$ and $a_2$ at their $home$ at some round $t$ where $t_0+2n+2 \le t \le t_0+3n+1$. In that case they find that there are 3 agents at their $home$ and  all of them change to state \textbf{Coloc} and executes algorithm \textsc{PerpExplore-Coloc-Pbl}.
\end{remark}

\section{Algorithm \textsc{PerpExplore-Scat-Whitbrd}}
\label{Appendix: scatWhitbrd}
In this section, we discuss the algorithm \textsc{PerpExplore-Scat-WhitBrd}, which achieves \textsc{PerpExploration-BBH}, with the help of three agents on a ring $R$ with $n$ nodes, where each node has a whiteboard that can store $O(\log n)$ bits of data. Note that, in this model as well we have considered that the agents are placed arbitrarily placed along the nodes of the ring $R$ (each such node is a `safe node'). So, there can be two cases, first, all the three agents are initially placed at three different positions, second, two agents are together whereas the third agent is in a different position. We will provide algorithm for the first case only, as an algorithm for the second case can be easily designed by  modifying and merging the algorithms for the first case and algorithm \textsc{PerpExplore-Coloc-Pbl} (as discussed in Remark~\ref{remark: Multiplicity case} in Appendix, Section~\ref{Appendix: multiplicity whitbrd case}).

\subsection{Detailed description and Pseudocode of the algorithm}
\label{Appendix: desc scat whitbrd}
Let us consider the three agents (say, $a_1$, $a_2$ and $a_3$) are initially placed at three nodes of the ring $R$, which are not only safe nodes but are also recognized as the $home$ of these agents. Initially, the agents starting from their respective $home$ are assigned the task to explore a set of nodes, which we term as a \textit{segment} of the corresponding agents. More precisely, a segment for an agent $a_i$ is defined as the set of consecutive nodes, starting from its $home$ and ending at the nearest clockwise $home$ (i.e., the $home$ of first clockwise placed agent), which is also termed as $Seg(a_i)$. Note that $\cup^{3}_{i=1}Seg(a_i) =R$. So if none of the agents are ever destroyed the ring will still be perpetually explored.

\noindent Let us first discuss the case, in which we describe the possible movements of the agent and the respective state changes they perform, until one agent gets destroyed by the Byzantine black hole, and another agent gets to know that an agent is already destroyed by the Byzantine black hole. After which, we will describe all the possible movements and state changes performed by the remaining two alive agents, between getting to know that already one has been destroyed by the Byzantine black hole, and finally detecting the Byzantine black hole's position.

Initially, all agents start from state \textbf{Initial} at their respective $home$. In state \textbf{Initial}, an agent ( without loss of generality, say $a_1$) first clear the already present data (if at all) at the whiteboards of their respective $home$, then initializes $T_{time}=0$ ($T_{time}$ stores the number of rounds elapsed since the start of state \textbf{Initial}), and writes a message of the form (\texttt{home, ID}) at its $home$, where \texttt{ID} is the ID of the agent. This type of message is termed as ``\texttt{home}" type message, which consists of two components, the first component stores the message \texttt{home} and the second component stores the ID of the agent writing, i.e., ID of the agent whose $home$ is the $Current-Node$. Further, it changes its state to \textbf{Forward} and moves in the clockwise direction. 

In state \textbf{Forward}, the agent moves in the clockwise direction while erasing the earlier direction marking (if exists), i.e., \texttt{left} and then writes the new direction marking, i.e., \texttt{right} in each node. The agent also increases the $T_{time}$ variable by one in each round. This process continues until the agent encounters a node that has a ``\texttt{home}" type message. This ``\texttt{home}" type message signifies that the agent has reached the end of segment $Seg(a_1)$, i.e., in other words, it has reached the nearest clockwise $home$, say $v_c$. Note that the length of a segment can be at most $n-2$, hence within $T_{time}=n-1$, an agent $a_1$ is bound to reach the last node of its own segment i.e., $v_c$. In any case irrespective of the current $T_{time}$, the agent waits at $v_c$ until $T_{time}=n-1$, after which in the next round, it checks for the following information at $v_c$. If the agent finds a message of type ``\texttt{visited}" at $v_c$, the agent considers this as an anomaly and learns that an agent of which $v_c$ is the home (from the ID component of \texttt{home} type message at $v_c$)  must have entered Byzantine black hole while returning back from its clockwise nearest $home$ (refer Lemma~\ref{lemma: agent sees visited type message at other agents home}). Then in this case, the agent stores the message of type \texttt{dir} in its local memory, where \texttt{dir=(Counter-clockwise, NULL)}. In a \texttt{dir} type message, the first component is called a \textit{direction component} which indicates the direction of the agent that gets destroyed by the Byzantine black hole, along which it was moving just before it gets destroyed. On the other hand, the second component either stores ID of some agent or stores $NULL$ message. Both these components are useful for certain state transitions. Next, the agent changes its state to \textbf{Gather}. Otherwise, if no anomaly is detected, then the agent simply writes the message \texttt{(Visited, self ID)} at $v_c$ and changes its state to \textbf{Back-Wait}.

Next, in state \textbf{Back-Wait}, an agent (here $a_1$) waits at $v_c$ until $T_{time}=2n-1$. In this waiting time, $a_1$ waits for the other agent to meet $a_1$ if the other agent detects any anomaly during its state \texttt{Forward}. While waiting, if it finds that the $\#Agent$ at $v_c$ is more than 1, and the whiteboard at $v_c$ has a message of type \texttt{dir}, then the agent performs the following task. After noticing this, the agent moves along the counter-clockwise direction, while changing its state to \textbf{Gather2}. Note that, our algorithm ensures that while in state \textbf{Back-Wait}, if a \texttt{dir} type message is seen by an agent, then the direction component of this message must be in a \texttt{counter-clockwise} direction (refer Lemma \ref{corollary: dir type counter-clockwise by back-wait}). Otherwise, if no such \texttt{dir} type message is seen by $a_1$, and also $T_{time}=2n$, then in this round, the agent changes its state to \textbf{Backtrack}.

In state \textbf{Backtrack}, an agent (here $a_1$) starts to move in a counter-clockwise direction from $v_c$, and after each move erases the earlier direction, i.e., \texttt{right}, and writes the new direction, i.e., \texttt{left}, and also increments $T_{time}$ by 1. This process continues until it reaches its own $home$, i.e., reads a \texttt{home} type message with ID same as its own ID. Again note that, since a segment can be of length at most $n-2$, hence within $T_{time}=3n-1$, the agent reaches its own $home$. After which it does not move until $T_{time}=3n$. At $T_{time}=3n$, it checks whether its $home$ has \texttt{visited} type message, if so then directly changes its state to \textbf{Initial-Wait}. Otherwise, the absence of such \texttt{visited} type message, creates an anomaly for the agent. This only happens when another agent, say $a_j$, ($j \ne 1$), for which $Seg(a_1) \cap Seg(a_j)=$ $home$ of $a_1$, does not arrive at $home$ of $a_1$ during its \textbf{Forward} state because of getting destroyed at the Byzantine black hole (refer Lemma \ref{lemma: sees no visited at bactrack}). This instigates the agent $a_1$ to change its state to \textbf{Gather} and store the \texttt{dir} type message \texttt{(clockwise, NULL)} in its local memory.

In state \textbf{Initial-Wait}, the agent waits at its $home$ until $T_{time}=4n-1$. This waiting period is enough for an anomaly finding agent in state \texttt{Backtrack}, to meet with it. If the agent finds that there is more than one agent at its $home$, and also there is a \texttt{dir} type message as well, then our algorithm ensures that the direction component of this \texttt{dir} type message must be in \texttt{clockwise} direction (refer Lemma \ref{corollary: dir type clockwise by initial-wait}). In this case, the agent starts moving clockwise and changes its state to \textbf{Gather1}. On the contrary, if $T_{time}=4n$ and no anomaly is detected, then the agent again moves back to state \textbf{Initial}.

An agent can change its state to \textbf{Gather}, either from state \textbf{Forward} or from state \textbf{Backtrack}. Note that in either case, it carries the stored message of type \texttt{dir} (which has a direction component and an ID component that is either $NULL$ or stores the ID of an agent). So, in this state if the direction component of \texttt{dir} is along \texttt{clockwise} (resp. \texttt{counter-clockwise}), then the agent moves along clockwise (resp. counter-clockwise) direction until it finds a node with $\#Agent$ more than one. Next, in case of a counter-clockwise direction component, the agent updates the \texttt{dir} message to (\texttt{Counter-clockwise, ID'}) where \texttt{ID'} is the ID of the other agent at the same node. Then the agent writes the updated \texttt{dir} message at the current node and changes state to \texttt{Gather2}. If the agent in the state \textbf{Gather} met the other agent when the \texttt{dir} type message has a clockwise direction component, the agent simply just writes the message at the current node and moves to state \textbf{Gather1}.

Also, note that exactly one agent can move into state \textbf{Gather}. This is because if an agent is destroyed during \textbf{Forward} state then the anomaly is first found by a single agent whose $home$ is the clockwise nearest $home$  of the destroyed agent. This agent changes its state to \textbf{Gather}. Now, before the other agent finds further anomaly the agent in state \textbf{Gather} meets the other agent and forces it to change into state \textbf{Gather1}. Similarly, If an agent $a_1$ is destroyed by the Byzantine black hole during its \textbf{Backtrack} state, then the first anomaly is detected by a single agent (more precisely the agent for which the clockwise nearest $home$ is the $home$ of $a_1$). In this case, it changes its state to \textbf{Gather} and forces the other agent to move into state \textbf{Gather2} by meeting it before it finds further anomaly. Thus we have the following observation.
\begin{observation}
There is exactly one agent that changes its state to \textbf{Gather} throughout the execution of Algorithm \textsc{PerpExplore-Scat-Whitbrd}. Also,
    The agent which changes its state to \textbf{Gather}, is the first agent to understand that, an agent has already been destroyed by the Byzantine black hole.
\end{observation}

An agent can change its state to \textbf{Gather1}, only from the states \textbf{Initial-Wait} or \textbf{Gather}.
An agent changes its state to \textbf{Gather1} from \textbf{Gather}, only when while moving along \texttt{clockwise} direction finds another agent present, in which case it writes the corresponding message of type \texttt{dir} and correspondingly changes its state to \textbf{Gather1}.
on the other hand, an agent (here $a_1$) can only change its state to \textbf{Gather1} from \textbf{Initial-Wait}, if it is waiting at its $home$, and within which it encounters another agent at its $home$ along with it a message of type \texttt{dir} is written (in which the direction component of \texttt{dir} is along \texttt{clockwise}). 
Note that if $a_1$ in state \textbf{Initial-Wait} finds more than one agent at the current node and \texttt{dir} type message at round $t$ then, the other agent must be in state \textbf{Gather} while written there at round $t-1$ and moved into state \textbf{Gather1} in the same round. Then, during round $t$, $a_1$ first stores the \texttt{dir} type message before it moves to state
\textbf{Gather1} and moves according to the direction component of the \texttt{dir} type message (i.e., clockwise). Also, the other agent that changed its state to \textbf{Gather1} from \textbf{Gather} in round $t-1$ also moves according to the \texttt{dir} type message during round $t$. Thus we have another observation
\begin{observation}
    \label{obs:gather1 moves together}
    If two agents move to state \textbf{Gather1}, one moves from state \textbf{Gather}, and the other moves from state \textbf{Initial-Wait}. Furthermore, they change to state \textbf{Gather1} at the same node and leave the node together at the same round and stay together.
\end{observation}

In state \textbf{Gather1}, the agent which changes its state to \textbf{Gather1} from \textbf{Initial-Wait} first stores the \texttt{dir} type message, as the other agent must already have stored this message in state \textbf{Gather}. Irrespective of which, the agent starts moving in a clockwise direction until it finds a node that has a \texttt{home} (i.e., of the form (\texttt{home}, \texttt{ID})) type message, where the ID component this message does not match with the IDs of the agents present at the current node. Note that our algorithm ensures that the current node is the $home$ of $a_j$ if $a_j$ is the agent to be destroyed by the Byzantine black hole (refer Lemma \ref{lemma: reachCautiousStartNode}). Further, as there are two agents at the current node both in state \textbf{Gather1} (Observation~\ref{obs:gather1 moves together}), so now the agent with the lowest ID, sets $Move=0$, also initializes the $Marking$ variable to $right$, since the direction component of \texttt{dir} is \texttt{clockwise}, and further changes its state to \textbf{Cautious-Leader}. On the other hand, the other alive agent at the current node changes its state to \textbf{Cautious-Follower}, while initializing $Move=0$.

\begin{algorithm} \label{algorithm: PerpExplore-Scat-WhitBrd}
\footnotesize
    \caption{\textsc{PerpExplore-Scat-WhitBrd}}
    \tcp{Algorithm is written for an agent $r$ }
    \textbf{Input:} $n$\\
    \textbf{States:}$\{\textbf{Initial}, \textbf{Forward}, \textbf{Back-Wait}, \textbf{BackTrack}, \textbf{Initial-Wait}, \textbf{Gather}$, \\ \hspace{1.16cm}$\textbf{Gather1}, \textbf{Gather2}, \textbf{Cautious-Leader}, \textbf{Cautious-Follower}\}$\\
    \underline{In State \textbf{Initial}:}\\
    Clear Whiteboard at the $Current-Node$\\
    $T_{time}=0$\\
    Write \texttt{(home, ID($r$))} at the $Current-Node$.\\
    Change state to \textbf{Forward} and move clockwise\\

    \underline{In State \textbf{Forward}:}\\
    \eIf{$T_{time} < n$}
    {   
        $T_{time}=T_{time}+1$\\
        \If{$Current-Node$ does not have any ``\texttt{home}'' type messages}
        {
            
            Write at the $Current-Node$ \texttt{right} and erase \texttt{left} (if exists) and move clockwise.\\
        }
        
    }
    {   
        $T_{time}=T_{time}+1$\\
        \eIf{$Current-Node$ already have a ``\texttt{visited}'' type message}
        {
            Store the message of type \texttt{dir}, where \texttt{dir}=(\texttt{Counter-clockwise},\texttt {NULL}) \\
            and then change to state \textbf{Gather}.\\
           
        }
        {
            Write \texttt{(Visited, ID($r$))} at the $Current-Node$ and change to state \textbf{Back-Wait}
        }
    }

    \underline{In State \textbf{Back-Wait}:}\\
    \eIf{$T_{time}<2n$}
    {
        $T_{time}=T_{time}+1$\\
        \If{$\#Agent$ at the $Current-Node>1$ $\wedge$ have a message of type ''\texttt{dir}"}
        {
        \If{direction component of \texttt{dir} message is of type ``\texttt{counter-clockwise}"}
        {
        Store the \texttt{dir} type message in local memory\\
        Move along direction counter-clockwise, and change to state \textbf{Gather2}.\\
        }
    
        }
        
    }
    {
    Move in counter-clockwise direction and change to state \textbf{Backtrack}.\\
    }
    \underline{In State \textbf{Backtrack}:}\\
    \eIf{$T_{time}<3n$}
    {
    $T_{time}=T_{time}+1$.\\
    \If{$Current-Node$ does not have any ``\texttt{home}" type message}
    {
    Write at the $Current-Node$ \texttt{left} and erase \texttt{right} and move counter-clockwise.\\
    }
    }
    {
    $T_{time}=T_{time}+1$\\
    \eIf{$Current-Node$ already have a ``\texttt{visited}" type message}
    {
    Change to state \textbf{Initial-Wait}.\\
    }
    {
    Store the message of type \texttt{dir}, where \texttt{dir}=(\texttt{Clockwise,NULL)} and then change to state \textbf{Gather}.\\ 
    }
    }
   \underline{In state \textbf{Initial-Wait}:}\\
   \eIf{$T_{time}<4n$}
   {
   $T_{time}=T_{time}+1$\\
   \If{$\#Agent$ at the $Current-Node>1$ $\wedge$ have a message of type ``\texttt{dir}"}
        {
        \If{direction component of \texttt{dir} message is of type ``\texttt{clockwise}"}
        {
        Store the \texttt{dir} type message in local memory\\
         Change to state \textbf{Gather1} and move clockwise.\\
        }
    
        }
   }
   {
   Change to state \textbf{Initial}.\\
   }
    \underline{In State \textbf{Gather}:}\\
   \eIf{direction component of \texttt{dir} is \texttt{clockwise}}
   {
   \eIf{$\#Agents$ at $Current-Node$ is 1}
   {
   Move in a clockwise direction.
   }
   {
   Write at $Current-Node$ \texttt{dir} and move to state \textbf{Gather1}.\\
   }
                                      
               }
   {
   \eIf{$\#Agents$ at $Current-Node$ is 1}
   {
   Move in a counter-clockwise direction.
   }
   {
   Update \texttt{dir}= \texttt{(counter-clockwise},\texttt{ID')} \tcp{\texttt{ID'} is the ID of the other agent} 
   Write at $Current-Node$ \texttt{dir} and move to state \textbf{Gather2}.\\
   }
   }
  
\end{algorithm}
\begin{algorithm}
\footnotesize
\setcounter{AlgoLine}{60}
    
   \underline{In State \textbf{Gather1}:}\\
   \eIf{$Current-Node$ has ``\texttt{home}" type message with the ID  component of the message does not match with IDs of the agent present at the $Current-Node$ \footnotemark}
   {
   \eIf{ID is the lowest among the set of IDs at the $Current-Node$}
   {
   Set $Move=0$.\\
   \eIf{direction component of \texttt{dir} is \texttt{clockwise}} 
   {Set $Marking= right $.}
   {
   $Marking= left$
   }
   Change its state to \textbf{Cautious-Leader}.\\
   }
   {
   Set $Move=0$\\
   Change to state \textbf{Cautious-Follower}.\\
   }
   }
   {
   Move in a clockwise direction.\\
   }
   \underline{In State \textbf{Gather2}:}\\
   \eIf{$Current-Node$ has a ``\texttt{home}" type message with ID component same as ID component of \texttt{dir}}
   {
   \eIf{ID is the lowest among the set of IDs at the $Current-Node$}
   {
   Set $Move=0$.\\
   \eIf{direction component of \texttt{dir} is \texttt{clockwise}} 
   {Set $Marking= right $.}
   {
   $Marking= left$
   }
   Change its state to \textbf{Cautious-Leader}.\\
   }
   {
    Set $Move=0$\\
   Change to state \textbf{Cautious-Follower}.\\
   }
   }
   {
   Move in a counter-clockwise direction.\\
   }
    \underline{In State \textbf{Cautious-Leader}:}\\
    \eIf{$Move=0$}
    {
        Update $Move=1$ and move according to the direction component of \texttt{dir}.\\
    }
    {
    \eIf{$Current-Node$ has exactly one agent}
    {
        \eIf{$Current-Node$ is already marked with $Marking$}
        {
            $Move=0$ and move opposite to the  direction component of \texttt{dir}
        }
        {
            Declare $Current-Node$ as the Byzantine black hole node and continue perpetual exploration avoiding this node.\\
        }
    
    }
    {
        Update $Move=0$ and moves according to the dirction component of \texttt{dir} type message in its local memory
    }
        
    }
    \underline{In State \textbf{Cautious-Follower}:}\\
    \eIf{$Move=0$}
    {
    Set $Wait=0$ and update $Move=1$.\\
    }
    {
    \If{$Wait<1$}
    {
    $Wait=Wait+1$.\\
    }
    {
    \eIf{$\#Agent$ at $Current-Node>1$}
    {
    Update $Move=0$ and move according to the direction component of \texttt{dir}.\\
    }
    {
     Declare the next node along the direction component of \texttt{dir} as the Byzantine black hole node and continue perpetual exploration avoiding this node.\\
    }
    }
    }
    
\end{algorithm}
\footnotetext{IDs present at the $Current-Node$ implies the IDs of the agents present at the $Current-Node$}

An agent can change its state to \textbf{Gather2}, only from the states \textbf{Back-Wait} or \textbf{Gather}. An agent changes its state to \textbf{Gather2} from \textbf{Gather}, only when while moving along \texttt{counter-clockwise} direction finds another agent present, in which case it updates the \texttt{dir} message and writes the corresponding message of type \texttt{dir} and correspondingly changes its state to \textbf{Gather2}.

On the other hand, an agent can only change its state to \textbf{Gather2} from \textbf{Back-Wait}, if it is waiting at its corresponding segments nearest clockwise $home$, and within which it encounters another agent at its current node along with it a message of type \texttt{dir} is written (in which the direction component of \texttt{dir} is along \texttt{counter-clockwise}). After storing the \texttt{dir} type message in its local memory and changing its state to \textbf{Gather2} from \textbf{Back-Wait}, an agent immediately moves counter-clockwise in the same round. This generates another observation similar to the Observation~\ref{obs:gather1 moves together} as follows.

\begin{observation}
\label{obs: gather2 together}
     If two agents move to state \textbf{Gather2}, one moves from state \textbf{Gather}, and the other moves from state \textbf{Back-Wait}. Furthermore, they change to state \textbf{Gather2} at the same node and leave the node together at the same round and stay together.
\end{observation}

In state \textbf{Gather2}, the agent moves in a counter-clockwise direction, until it finds a node which has a \texttt{home} (i.e., of the form (\texttt{home}, \texttt{ID'})) type message, where \texttt{ID'} matches with ID component of \texttt{dir} type message in its local memory. Note that our algorithm ensures that the current node is the nearest clockwise $home$ of $a_1$ if $a_1$ is the agent to be destroyed by the Byzantine black hole (refer Lemma \ref{lemma: reachCautiousStartNode}). Further, as there are two agents at the current node, both in state \textbf{Gather2} (refer Observation~\ref{obs: gather2 together}), so now the agent with the lowest ID, sets $Move=0$, also initializes the $Marking$ variable to $left$, since the direction component of \texttt{dir} is \texttt{counter-clockwise}, and further changes its state to \textbf{Cautious-Leader}. On the other hand, the other alive agent at the current node, changes its state to \textbf{Cautious-Follower}, while initializing $Move=0$.

In state \textbf{Cautious-Leader}, an agent updates $Move=1$ and moves along the direction component of \texttt{dir}. After moving when it is alone on a node, it checks if the current node is marked with $Marking$ (i.e., either $left$ or $right$). If the node is marked, then it moves in the opposite direction (i.e., the opposite of direction component of \texttt{dir} message) to meet with the agent in state \textbf{Cautious-Follower}. Otherwise, if there is no $Marking$ at the current node then it identifies the current node as the Byzantine black hole node (only if it remains alive) and continues perpetual exploration avoiding this node. If an agent in the state \textbf{Cautious-Leader} with $Move=1$ meets with another agent (i.e., the agent in the state \textbf{Cautious-Follower}) it updates variable $Move$ to 0 and again moves according to the direction component.

In state \textbf{Cautious-Follower}, the agent after finding $Move=0$, initializes the variable $Wait=0$ and updates $Move=1$. If $Move=1$, then the agent checks whether $Wait<1$, if so then update $Wait=Wait+1$. When $Wait=1$, it checks if the current node has more than one agent (i.e., whether the agent in state \textbf{Cautious-Leader} has returned or not), if so then further update $Move=0$ and moves along the direction component of \texttt{dir}, along with the agent in state \textbf{Cautious-Leader}. Otherwise, if the agent in state \textbf{Cautious-Leader} does not return, that means the number of agents at the current node is 1 when $Wait=1$, this symbolizes that the agent in state \textbf{Cautious-Leader} has been either destroyed by the Byzantine black hole or the agent in state \textbf{Cautious-Leader} finds that its current node (i.e., the next node along the direction component of \texttt{dir} for the agent in state \textbf{Cautious-Follower}) is the Byzantine black hole and continues to perpetually move in the same direction, further avoiding this node. In any case, the agent executing \textbf{Cautious-Leader} has not returned implies the next node along the direction component of \texttt{dir} for the agent executing \textbf{Cautious-Follower} is the Byzantine black hole node. So, in this case, the agent declares the next node along the direction component of \texttt{dir} is the Byzantine black hole node and continues perpetual exploration avoiding that node.

\subsection{Correctness and Complexity}
\label{Appendix:correct whitbrd scat}
In the next two lemmas (Lemma~\ref{lemma: agent sees visited type message at other agents home} and Lemma~\ref{lemma: sees no visited at bactrack}) we first ensure that when we say an agent encounters an anomaly, is actually an anomaly. And what an agents interpret from these anomalies are true. Next in Lemma~\ref{lemma: gather gets actual direction} we ensure that the agent that changes to state \textbf{Gather} can always identify the state of the destroyed agent at the time it gets destroyed at the Byzantine black hole.

\begin{lemma}\label{lemma: agent sees visited type message at other agents home}
    Let $a_i$ and $a_j$ be two agents exploring the segments $Seg(a_i)$ and $Seg(a_j)$, such that $Seg(a_i)\cap Seg(a_j)=v_c$, where $v_c$ is also the $home$ of $a_j$. If during the execution of \textsc{PerpExplore-Scat-Whitbrd}, there exists a round $t>0$, in which $a_i$ in state \textbf{Forward} finds a \texttt{visited} type message, then there exists a round $0<t'<t$, in which $a_j$ has been destroyed by the Byzantine black hole while exploring the segment $Seg(a_j)$ in state \textbf{Backtrack}. 
\end{lemma}

\begin{proof}
    Let us suppose, there does not exist any round $t'<t$ in which $a_j$ has been destroyed by the Byzantine black hole in state \textbf{Backtrack}. This means that either $a_j$ is yet to be destroyed by the Byzantine black hole at round $t'$, for all $t'<t$ or $a_j$ has been destroyed by the Byzantine black hole at round $t'$, while it is in state \textbf{Forward}.\\
    \noindent \textit{\underline{Case I:}} Let $a_j$ is yet to be destroyed by the Byzantine black hole at round $t'$, for all $t'<t$. Note that since $a_i$ is at $v_c$ and checks for the \texttt{visited} type message, this implies $T_{time}=n$. This means at $(t-n-1)$-th, round $a_i$ changed its state to \textbf{Forward} from \textbf{Initial}. This also implies $a_j$ was also in state \textbf{Initial} at $(t-n-1)-$th round at $v_c$. Hence, at this round, $a_j$ must have cleared any whiteboard data at $v_c$ (as $v_c$ is the $home$ of $a_j$) and also has moved clockwise after changing the state to \textbf{Forward}. Observe that, the node $v_c$ can only be explored by $a_i$ and $a_j$. So, these two agents have the possibility to write the \texttt{visited} type message at $v_c$ after round $t-n-1$ and before round $t$. Further, note that in the next $n+1$ rounds after $(t-n-1)-$th round (i.e., until round $t$), $a_j$ cannot return back to $v_c$. So, $a_j$ can not write anything on $v_c$ after round $t-n-1$ and before round $t$. Also, $a_i$ can only write the \texttt{visited} type message at $v_c$ at the $t-$th round (i.e., when $T_{time}=n$). But, in this case, $a_i$ at $t-$th round before writing a message, finds that there already exists a \texttt{visited} type message, and this a contradiction, to the fact that $a_j$ has not been destroyed by the Byzantine black hole at round $t'$, for all $t'<t$. \\
    \noindent \textit{\underline{Case II:}} Let $a_j$ has been destroyed by the Byzantine black hole at some round $t'<t$, while in state \textbf{Forward}. 
    This implies there exists a round $0<t'' <t'$ in which $a_j$ was in state \textbf{Initial} 
    (this is the last time $a_j$ was in state \textbf{Initial} before it gets destroyed by the Byzantine black hole). Thus this means, that $t' \le t''+n-2$, as a segment can be of length at most $n-2$. Then at round $t''$, $a_i$ was also in state \textbf{Initial} at its own $home$. Now, note that after $t''$, $a_i$  can visit $v_c$ within round $t''+1$ and $t''+n+1$  while in state \textbf{Forward}. Note that from round $t''+1$ and before round $t''+n+1$ no agent can write a \texttt{visited} type message at $v_c$ (by a similar argument as in Case I). So during round $t''+n+1$, $a_i$ can not see any \texttt{visited} type message at $v_c$. So, $t > t''+n+1$. Also note that the next time after round $t''+n+1$, $a_i$ visits $v_c$ in state \textbf{Forward}, is earliest at round $t''+4n+1$. This implies $t \ge t''+4n+1$. Now, let us consider the agent $a_k$ where $Seg(a_k) \cap Seg(a_j) = home$ of $a_k$ and $Seg(a_j)\cap Seg(a_i)=home$ of $a_i$. Note that in round $t''+n+1$ both $a_k$ and $a_i$ move into state \textbf{Back-Wait} (as none of them sees any \texttt{visited}  type message). Next after waiting there up to round $t''+2n$, both of them change their state to \textbf{Backtrack} at round $t''+2n+1$. Next, they reach their corresponding $home$ and check for a \texttt{visited} type message at round $t''+3n+1$. Note that $a_i$ finds the \texttt{visited} type message left by $a_k$ and changes to state \textbf{Initial-Wait} whereas, $a_k$ does not find any \texttt{visited} type message left by $a_j$ (as $a_j$ is destroyed at the Byzantine black hole at state \textbf{Forward} at round $t' \le t''+n-2 < t''+3n+1$). So, $a_k$ changes to state \textbf{Gather1} at round $t''+3n+1$ and moves clockwise with the message \texttt{dir}$= $ \texttt{(Clockwise, NULL)} until it finds $a_i$ at it's $home$ in state \textbf{Initial-Wait} before round $t''+4n$ (as length of a segment can be of length at most $n-2$). Thus before $t''+4n+1$-th round $a_i$ changes its state to \textbf{Gather1}. Since no agent moves to state \textbf{Forward} from state \textbf{Gather1}, $a_i$ can not be at $v_c$ in state \textbf{Forward} during round $t$. So we again arrive at a contradiction. This implies $a_j$ can not be in state \textbf{Forward} while it is destroyed at the Byzantine black hole at round $t'<t$. \end{proof}

\begin{lemma}
\label{lemma: sees no visited at bactrack}
    Let $a_i$ and $a_j$ be two agents such that $v_c= Seg(a_i)\cap Seg(a_j)= home$ of $a_i$. If there exists a round $t >0$ such that $a_i$ checks and finds no \texttt{visited} type message on $v_c$ while in state \textbf{Backtrack}, then there must be a round $t'<t$ in which $a_j$ was destroyed by the Byzantine black hole while in state \textbf{Forward}.
\end{lemma}

\begin{proof}
    Let us consider that there does not exist any round $t'<t$ where $a_j$ was destroyed by the Byzantine black hole in the state \textbf{Forward}. This implies that either $a_j$ is not destroyed by the Byzantine black hole for all rounds  $t'<t$ or it is destroyed by the Byzantine black hole at round $t'<t$ while it was in state \textbf{Backtrack}.
    
    \noindent\textit{\underline{Case-I:}} Let us consider that $a_j$ is yet to be destroyed by the Byzantine black hole at round $t',~\forall t>t'>0$. Note that, at round $t$, $a_i$ checks for a \texttt{visited} type message while in state \textbf{Backtrack}, and that means the current $T_{time}$ for any alive agent must be equal $3n$ (at $T_{time}=n$ every alive agent ends its state \textbf{Forward}, then until $T_{time}=2n$ every such agent is in state \textbf{Back-Wait}, and at $T_{time}=3n$ every alive agent checks for the message in state \textbf{Backtrack}). Note that, this means at time $t-3n-1$, both $a_i$ and $a_j$ were in state \textbf{Initial} at their respective $home$. Then at round $t-2n-1$, both of them must be in state \textbf{Forward} with $T_{time}=2n$, while $a_j$ is at $v_c$ and $a_i$ is at $home$ of $a_k$ ( $a_k$ is the third agent exploring $Seg(a_k)$, such that $Seg(a_i)\cap Seg(a_k)=home$ of $a_k$ and $Seg(a_k)\cap Seg(a_j)=home$ of $a_j$). So, at this point, both $a_i$ and $a_j$ must have written a \texttt{visited} type message at their respective current nodes (i.e., $v_c$ and $home$ of $a_k$). Further observe that, within $t-2n$ round to $t-n-1$ round $a_j$ remains at $v_c$ while in state \textbf{Back-Wait}, and according to our algorithm no agent in state \textbf{Back-Wait} can alter the already stored information at their current node, i.e., within these rounds, $a_j$ cannot erase the \texttt{visited} type message at $v_c$. Next, within $t-n$ to $t-1$ round, the only agent that can visit the node $v_c$ is $a_i$, while it is in state \textbf{Backtrack}. Note that, $a_i$ cannot alter any \texttt{visited} type message during state \textbf{Backtrack}. Thus, $a_i$ at round $t$ checks and finds a \texttt{visited} type message at $v_c$. This is a contradiction to the assumption that $a_i$ does not find any \texttt{visited} type message at round $t$. Thus, $a_j$ must have been destroyed by the Byzantine black hole at time $t'$ for some $t'<t$.
\vspace{0.4cm}

    \noindent\textit{\underline{Case-II:}} Suppose $a_j$ has been destroyed by the Byzantine black hole at some round $t'<t$, while in state \textbf{Backtrack}. Let $t''$ be the round when $a_j$ was in state \textbf{Initial} for the last time where $0<t''<t'$. Note that $t'> t''+2n+1$ (because an agent changes its state to \textbf{Backtrack} at round $t''+2n+1$). So at the round, $t''+n+1$ it was alive and at $v_c$ in state \textbf{Forward}. During this round, it also writes a \texttt{visited} type message at $v_c$. Now as we discussed earlier in Case-I, from round $t''+n+1$ till round $t-1$ no agent can erase or alter the \texttt{visited} type message at $v_c$. Thus at round $t$, $a_i$ must find the \texttt{visited} type message at $v_c$ upon checking while in state \textbf{Backtrack}. This is a contradiction to the fact that at round $t$, $a_i$ checks and finds no \texttt{visited} type message at $v_c$ while it is in state \textbf{Backtrack}. Thus if $a_j$ is destroyed by the Byzantine black hole at round $t'$, then it must be in state \textbf{Forward}.
\end{proof}

\begin{lemma}
\label{lemma: gather gets actual direction}
    The agent which changes its state to \textbf{Gather}, correctly identifies the state in which an agent was in, just before it gets destroyed by the Byzantine black hole. 
\end{lemma}

\begin{proof}
    Let $a_i,a_j$, and $a_k$ be three agents initially at three different nodes (i.e., three different $home$ for three agents) of $R$. Let $Seg(a_i)\cap Seg(a_j)=home$ of $a_i$, $Seg(a_i)\cap Seg(a_k)=home$ of $a_k$ and $Seg(a_k)\cap Seg(a_j)=home$ of $a_j$. Without loss of generality let $a_j$ be the agent that is destroyed by the Byzantine black hole at a round $t>0$. Now we have two cases according to the state of $a_j$ at the time of getting destroyed.
    
    \noindent\textit{\underline{Case-I:}} $a_j$ gets destroyed by the Byzantine black hole while in state \textbf{Forward} at round $t$. Let $t'<t$ be the round when $a_j$ was in state \textbf{Initial} for the last time. Note that at round $t'$, all agents are in state \textbf{Initial}. We claim that $a_i$ is the agent to change its state to \textbf{Gather}. If possible let $a_k$ change its state to \texttt{Gather}. Then first we show that it must change its state to \textbf{Gather} before round $t'+3n+1$. Otherwise, since $a_j$, fails to reach $home$ of $a_i$ at round $t'+n+1$, it does not write any \texttt{visited} type message there but $a_i$ reaches $home$ of $a_k$ at round $t'+n+1$ and writes a \texttt{visited} type message at $home$ of $a_k$. Also, since, $a_i$ and $a_j$ are the only two agents which visit $home$ of $a_i$., this implies, no message can be written at $home$ of $a_i$ on and between rounds $t'$ and $t'+3n$. So, when $a_i$ returns $home$ and finds no \texttt{visited} type message round $t'+3n+1$ while in state \textbf{Backtrack}, it changes its state to \textbf{Gather} but $a_k$ doesn't do so as it sees the \texttt{visited} type message at its $home$, left by $a_i$. So, if $a_k$ is the agent to change state to \textbf{Gather} it must be before round $t'+3n+1$. Thus $a_k$ must move into state \textbf{Gather} from state \textbf{Forward} at round $t'+n+1$. This can only occur if $a_k$ finds a \texttt{visited} type message at $home$ of $a_j$ at round $t'+n+1$. But this is not possible as $a_j$ have already erased all previous data on its $home$ at round $t'$ while in state \textbf{Initial}, and there is no other agent which can visit $home$ of $a_j$ and alter the data between $t'$ and $t'+n+1$ rounds. Hence, at round $t'+n+1$ when $a_k$ visits the $home$ of $a_j$ it does not find any \texttt{visited} type message. Thus $a_i$ must be the agent to change its state to \textbf{Gather}. Now to prove that $a_i$ stores the \textit{dir} type message with direction component clockwise  (which means $a_j$ was in state \textbf{Forward} while it was destroyed), we have to show that $a_i$ changes its state to \textbf{Gather} from \textbf{Backtrack}. If possible let $a_i$ change its state to \textbf{Gather} from state \textbf{Forward}. Then it must be at round $t'+n+1$ when $a_i$ checks and finds a visited type message at the $home$ of $a_k$. Now at round $t'$, $a_k$ which was in state \textbf{Initial}, cleared all previous data on its $home$. This implies the \texttt{visited} type message that $a_i$ finds at round $t'+n+1$ must be written there, after round $t'$. But as only $a_i$ can be there after $t'$ and before $t'+n+1$, and since it can not alter any data on $home$ of $a_k$ before round $t'+n+1$, it finds no \texttt{visited} type message at $home $ of $a_k$ at round $t'+n+1$. Thus $a_i$ can not change its state to \textbf{Gather} from state \textbf{Forward}. Hence, it must change its state to \textbf{Gather} from state \textbf{Backtrack} at round $t'+3n+1$ and according to the algorithm \textsc{PerpExplore-Scat-WhitBrd}, the \texttt{dir} type message that $a_i$ stores must have direction component \texttt{clockwise}. 
    

    \noindent \textit{\underline{Case-II:}} Let $a_j$ gets destroyed by the Byzantine black hole at round $t$ while in state \textbf{Backtrack}. We have to show that the agent that changes the state to \textbf{Gather} must store the \texttt{dir} type message with direction component \texttt{counter-clockwise} (as any agent can only move in the counter-clockwise direction, in state \textbf{Backtrack}). It will be enough to show that the agent that changes its state to \textbf{Gather} must have changed it from state \textbf{Forward} (as only in this case the state changing agent stores the \texttt{dir} type message having a counter-clockwise direction component in its local memory). Let $t'$ be the round when $a_j$ was in state \textbf{Initial} for the last time. We claim that $a_i$ and $a_k$ can not change to state \textbf{Gather} before round $t'+3n+1$. Note that at round $t'$ both $a_i$ and $a_k$ are at their corresponding $home$ in the state \textbf{Initial}. Now if any one of them changes to state \textbf{Gather}, the earliest it can happen is at round $t'+n+1$ when both of them are in state \textbf{Forward}. In this case, the agent that changes to state \textbf{Gather} must have found a \texttt{visited} type message at the current node (i.e., for $a_i$ it is $home$ of $a_k$ and for $a_k$ it is $home$ of $a_j$) at round $t'+n+1$. Note that during round $t'$ any previous messages are erased from both $home$ of $a_k$ and $home$ of $a_j$ by $a_k$ and $a_j$, respectively, and no other agent can visit and alter data at these nodes before $t'+n+1$ round. Hence none of $a_i$ and $a_k$ finds any \texttt{visited} type message at their current nodes at round $t'+n+1$. Hence, none of these agents changes their state to \textbf{Gather} at round $t'+n+1$. Next, they can only change their state to \textbf{Gather} at the round $t'+3n+1$ when both of them (i.e., $a_i$ and $a_k$) are at their respective $home$. An agent among them changes its state to \textbf{Gather} at round $t'+3n+1$ if it finds no \texttt{visited} type message at their current node (i.e., corresponding $home$). Note that $t'+3n+1>t>t'+n+1$ (as $a_j$ was in state \textbf{Backtrack} at round $t$). So, at the round $t'+n+1$ all of $a_i,a_j$ and $a_k$ are at nodes $home$ of $a_k, home$ of $a_i$ and $home$ of $a_j$ in state \textbf{Forward} and writes a \texttt{visited} type message in the nodes, respectively. These messages can not be altered by any agent until round $t'+3n+1$ (by a similar argument as in Case-I). Note that at round $t'+3n+1$ both $a_i$ and $a_k$ are at their corresponding $home$ and both of them find a \texttt{visited} type message at these nodes left by $a_j$ and $a_i$ respectively (at round $t'+n+1$). So, none of them changes to state \textbf{Gather} even at round $t'+3n+1$. Next, they move into state \textbf{Initial-Wait} and wait at their $home$ until $t'+4n$ round. Now, since no agent in state \textbf{Initial-Wait} can change its state to \textbf{Gather2} until it meets with an agent in state \textbf{Gather} and all alive agents at round $t'+3n+2$ are at state \textbf{Initial-Wait} all of them (i.e., $a_i$ and $a_k$) waits and changes state to \textbf{Initial} again at round $t'+4n+1$. Next at round $t'+4n+2$ both $a_i$ and $a_k$ erase all previous data at their corresponding $home$. and changes status to \textbf{Forward} again. Note that the next time $a_i$ and $a_k$ can change the state to gather must be at the round $t'+5n+2$ when both of them are in state \textbf{Forward}, and are currently at the $home$ of $a_k$ and $home$ of $a_j$, respectively. Further, note that since $home$ of $a_k$ does not have any \texttt{visited} type message at round $t'+5n+2$ (as $a_k$ erased any data at round $t'+4n+2$ and no other agent can alter data there, after $t'+4n+2$ and before $t'+5n+2$), so $a_i$ can not change to state \textbf{Gather} at round $t'+5n+2$. Also observe that the \texttt{visited} type message at the $home$ of $a_j$ written by $a_k$ during round $t'+3n+1$ is still there at round $t'+5n+2$ (after $t'+3n+1$ before $t'+4n+2$ since $a_j$ is already destroyed before $t'+3n+1$, no agent is on $home$ of $a_j$ to erase the data, also after $t'+4n+1$ till $t'+5n+1$ only $a_k$ can visit $home$ of $a_j$ in state \textbf{Forward} and \textbf{Back-Wait}, but in these states it does not alter any data at the $home$ of $a_j$) $a_k$ finds it during the round $t'+5n+2$ in state \textbf{Forward} and changes its state to \textbf{Gather}.
\end{proof}

    \begin{lemma}
        \label{corollary: dir type counter-clockwise by back-wait}
If an agent finds a \texttt{dir} type message while it is in state \textbf{Back-Wait}, the direction component of this message must be counter-clockwise.
    \end{lemma}
    \begin{proof}
        Suppose an agent $a_i$ in state \textbf{Back-Wait} gets a \texttt{dir} type message with clockwise
 direction component at some round $t$. This implies there exists another agent $a_k$ that has changed its state to \textbf{Gather} after storing a \texttt{dir} type message having a direction component clockwise at some round $t'<t$. This can only happen if $a_k$ was in state \textbf{Backtrack} at the beginning of round $t'$. So at the beginning of round $t'$, $a_i$ was also in state \textbf{Backtrack}, and thus at round $t'$, $a_i$ changes its state to \textbf{Initial-Wait}. Note that $a_k$ meets and shares \texttt{dir} type message with $a_i$ while $a_i$ is still at state \textbf{Initial-Wait}. This contradicts our assumption that $a_i$ gets \texttt{dir} message at state \textbf{Back-Wait}. Thus, if $a_i$ finds a \texttt{dir} type message in the state \textbf{Back-Wait} then it must have the direction component counter-clockwise. \end{proof}
With a similar argument, we can also prove the following lemma.
  \begin{lemma}\label{corollary: dir type clockwise by initial-wait}
      If an agent finds a \texttt{dir} type message while it is in state \textbf{Initial-Wait}, the direction component of this message must be clockwise.
  \end{lemma}

\begin{definition}[Cautious start node]
    Let $a_i$, $a_j$, and $a_k$ be three agents executing the algorithm \textsc{PerpExplore-Scat-Whitbrd}, and suppose $a_j$ be the first agent to enter the Byzantine black hole, while exploring the segment $Seg(a_j)$. Let $v_1$ be the $home$ of $a_j$ and $v_2$ be the furthest node from $v_1$ along the clockwise direction which is inside $Seg(a_j)$. We define the Cautious start node to be $v_1$ if $a_j$ is destroyed by the Byzantine black hole during state \textbf{Forward}. Otherwise, if $a_j$ is destroyed by the Byzantine black hole in state \textbf{Backtrack} then, $v_2$ is defined to be the Cautious start node.
\end{definition}

\begin{lemma}\label{lemma: reachCautiousStartNode}
 Let the first agent be destroyed by the Byzantine black hole at some round $t>0$, then there exists a round $t'>t$, at which the remaining alive agents reach the cautious start node and change their state to \textbf{Cautious-Leader} and \textbf{Cautious-Follower}.
\end{lemma}

\begin{proof}
    Let $a_i$, $a_j$ and $a_k$ be three agents, exploring $Seg(a_i)$, $Seg(a_j)$ and $Seg(a_k)$, respectively, where $Seg(a_i)\cap Seg(a_j)=home$ of $a_i$, $Seg(a_i)\cap Seg(a_k)=home$ of $a_k$ and $Seg(a_k)\cap Seg(a_j)=home$ of $a_j$. Let $a_j$ be the agent that gets destroyed by the Byzantine black hole while exploring $Seg(a_j)$, and then we have the following cases:

    \noindent\textit{\underline{Case-I:}} $a_j$ is destroyed by the Byzantine black hole while it is moving in a clockwise direction, i.e., in state \textbf{Forward} at round $t>0$. Note that in this case, the cautious start node is the $home$ of $a_j$. This implies there exists a round $0<t_0<t<t_0+n$, where $a_j$ was in state \textbf{Initial}. Note that, $a_j$ can not write any \texttt{visited} type message at $home$ of $a_i$, as it gets destroyed before reaching that node. So, at round, $t_0+3n+1$, when $a_i$ is at its own $home$ in state \textbf{Backtrack} and checks any \texttt{visited} type message, it finds none exists. This triggers $a_i$ to change its state to \textbf{Gather} from \textbf{Backtrack} with the corresponding \texttt{dir} = (\texttt{clockwise, NULL}) (which is indeed the correct direction, refer Lemma \ref{lemma: gather gets actual direction}). Further, the agent finds $a_k$, which is currently waiting at its own $home$ in state \textbf{Initial-Wait} (as it does not find any anomaly, so from \textbf{Backtrack} it changed to state \textbf{Initial-Wait} at round $t_0+3n+1$). So, the moment $a_i$ reaches the $home$ of $a_k$, it takes one additional round to store the \texttt{dir} type message at the current node and then changes its state to \textbf{Gather1}. On the other hand, whenever $a_k$ finds a \texttt{dir} type message is written at its current node, it also changes its state to \textbf{Gather1}, i.e., at round $t_0+3n+1<t''<t_0+4n$ both agents change their state to \textbf{Gather1}. After which, they together start moving in a clockwise direction (\ref{obs:gather1 moves together}), until they reach a node which has a \texttt{home} type message with the ID component, different from the IDs of $a_i$ and $a_k$. Note that this must be the $home$ of $a_j$, as the \texttt{home} type message $a_j$ has written at round $t_0$, cannot be erased by any other agent except $a_j$. Also $a_j$ can only erase this in state \textbf{Initial} at round $t_0+4n+2$ if it would have returned back, but since it is already destroyed between round $t_0+1$ and $t_0+n-1$, hence this possibility never arises, so the \texttt{home} type message remains, when $a_i$ and $a_k$ together reaches this node. After which they change their state to \textbf{Cautious-Leader} and \textbf{Cautious-Follower} depending on their IDs.

    \noindent\textit{\underline{Case-II:}} $a_j$ is destroyed by the Byzantine black hole while it is moving in a counter-clockwise direction, i.e., in state \textbf{Backtrack} at round $t>0$. Note that in this case, the cautious start node is the $home$ of $a_i$. Let $t_0$ be the round when $a_j$ was in state \textbf{Initial} the last time. This implies $t_0+2n<t<t_0+3n$, now by similar argument as explained in Case-II of Lemma \ref{lemma: gather gets actual direction}, $a_k$ changes its state to \textbf{Gather} while storing the \texttt{dir} = (\texttt{counter-clockwise, NULL}) message, at the $home$ of $a_j$ from state \textbf{Forward}, and starts moving in a counter-clockwise direction. Note that at round $t_0+5n+2$ as $a_i$ did not find any anomaly, so it changes its state to \textbf{Back-Wait} from \textbf{Forward}. Hence,  at round $t_1$, where  $t_1<t_0+6n$, $a_k$ finds $a_i$, while $a_i$ is still in state \textbf{Back-Wait}. This triggers $a_k$ to change its state to \textbf{Gather2} at round $t_1$ while updating the \texttt{dir} type message to (\texttt{Counter-clockwise, ID'}), where \texttt{ID'} is the ID of $a_i$. Then at the same round, $a_k$ writes the updated message at the current node (i.e., $home$ of $a_k$). Whenever $a_i$ sees this \texttt{dir} type message (at round $t_1+1$) it also changes its state to \textbf{Gather2}. Next in state \textbf{Gather2}, both start to move counter-clockwise (from round $t_1+2$) and continue to move until they find a \texttt{home} type message with ID matching the ID of the \texttt{dir} type message. Note that, this node is nothing but the $home$ of $a_i$ (as the ID component of \texttt{dir} type message stores the ID of $a_i$). Note that at round $t_0+4n+2$, $a_i$ written a \texttt{home} type message at its own $home$. This message can be erased only again at round $t_0+8n+2$ (i.e., when $a_i$ reaches its $home$ again in state \texttt{Initial}). But since $a_i$ changed its state to \textbf{Gather2} before $t_0+6n+1$ it cannot move back to state \textbf{Initial} again. So,when $a_k$ and $a_i$ reaches $home$ of $a_i$, they finds the \texttt{home} type message. Hence, both $a_k$ and $a_i$ reach the cautious start node within $t_0+7n$ and further change their states to \textbf{Cautious-Leader} and \textbf{Cautious-Follower}, depending on their IDs. \end{proof}

\begin{lemma}\label{lemma: CautiousLeaderCautiousFollower}
    Let $a_i$ and $a_k$ be the two agents that start the states \textbf{Cautious-Leader} and \textbf{Cautious-Follower} from the cautious start node, then within finite rounds of executing algorithm \textsc{PerpExplore-Scat-WhitBrd}, at least one agent detects the location of the Byzantine black hole.
\end{lemma}

\begin{proof}
    Let $a_j$ be the agent that has been destroyed by the Byzantine black hole at round $t>0$. Now there are two cases based on the state of $a_j$ at round $t$.

    \noindent \textit{\underline{Case-I:}} $a_j$ was in state \textbf{Forward} at round $t$. In that case, the node $v_1= home$ of $a_j$ is the cautious start node. Let $v_2$ be the farthest node along the clockwise direction in $Seg(a_j)$ (i.e., the clockwise nearest $home$ of $a_j$). Now if there exists any round $t'<t$ when $a_j$ was in the state \textbf{Backtrack} then it must have started its state \textbf{Backtrack} from $v_2$ and moved counter-clockwise until $v_1$ while erasing all \texttt{right} markings from each node in $Seg(a_j)$. So when $a_j$ started the state \textbf{Forward} at its $home$ before being destroyed, all nodes in $Seg(a_j)$ are without any \texttt{right} marking. This case will happen also when there is no round $t'<t$ where $a_j$ was in state \textbf{Backtrack}, i.e., there does not exist any round before $t'$, when $a_j$ was in state \textbf{Forward}. So before $a_j$ is destroyed at the Byzantine black hole, say $v_b$, at round $t$, it has marked  \texttt{right} at all nodes starting from the next node of $v_1$  in the clockwise direction, up to the node just before $v_b$ ($v_b$ can not be marked as before marking it the agent $a_j$ is destroyed there). Let without loss of generality, $a_i$ is in state \textbf{Cautious-Leader} and $a_j$ is in state \textbf{Cautious-Follower} at $v_1$ at some round $t_0>t$. Then $a_i$ always moves ahead alone in the clockwise direction to a new node $v$.Then it moves back only if sees the \texttt{right} marking and brings $a_k$ to $v$ along with it. Note that $a_i$ always sees a \texttt{right} marking until $v_b$. Now when it moves to $v_b$, if it is not destroyed it must see no such \texttt{right} marking as $a_j$ failed to mark it at round $t$. In this case, $a_i$ leaves the node in the clockwise direction to a new node. So here $a_i$ is able to detect the Byzantine black hole. On the other hand, if $a_i$ is destroyed while it visits $v_b$, then it can not return back to $a_k$. When $a_k$ sees that $a_i$ has not returned it interprets that $a_i$ must have been destroyed at the Byzantine black hole which is the next node along the clockwise direction. Thus in this case also at least one agent can detect the Byzantine black hole.

    \noindent\textit{\underline{Case-II:}} $a_j$ was in state \textbf{Backtrack} at round $t$. In this situation, let without loss of generality, the node $v_2$ be the $home$ of $a_i$ and, it is also the \textit{cautious start node}. Let $v_1$ be the farthest node along counter-clockwise direction in $Seg(a_j)$ (i.e., the $home$ of $a_j$). Now, if there exists any round $t'<t$ when $a_j$ was in state \textbf{Forward}, then it must have started this state from $v_1$ and moved in a clockwise direction until it reaches the node $v_2$ while erasing all the \texttt{left} markings from each node it traverses in $Seg(a_j)$. This means that, when $a_j$ started the state \textbf{Backtrack} from $v_2$, there does not exist any node in $Seg(a_j)$ with \texttt{left} markings. So before $a_j$ is destroyed by the Byzantine black hole node $v_b$ (say) at round $t$, it must have marked all the nodes with \texttt{left}, starting from the next counter-clockwise node of $v_2$ to the adjacent clockwise neighbor of $v_b$ (as before writing this message at $v_b$, the agent gets destroyed by the Byzantine black hole). Let without loss of generality, $a_i$ be the lowest ID agent among $a_i$ and $a_k$, hence it starts in state \textbf{Cautious-Leader}, whereas $a_k$ starts in state \textbf{Cautious-Follower}, at some round $t_0>t$. This means $a_i$ is the first agent to move alone in the next counter-clockwise neighbor say, $v$. After which, only if it sees a \texttt{left} message then only it moves back in the clockwise direction at the node of $a_k$, and in the next round both these agents reach the node $v$. Observe that, $a_i$ always finds a \texttt{left} message until the node $v_b$. Whenever it reaches $v_b$, either it gets destroyed by the Byzantine  black hole, or it finds that no \texttt{left} marking is present at the current node. This triggers $a_i$ to detect the current node to be the Byzantine black hole and moves in the counter-clockwise direction to a new node. Otherwise, if it also gets destroyed by the Byzantine black hole, then in the next round it is unable to return to $a_k$, which triggers $a_k$ to conclude that $a_i$ must have been destroyed by the Byzantine black hole as well, and it correctly detects the Byzantine black hole to be the next node in the counter-clockwise direction. Thus for each scenario, there exists at least one agent that is able to correctly detect the Byzantine black hole location.
\end{proof}

  Note that within at most $3n$ number of rounds after both alive agent starts cautious walk from the cautious start node, the Byzantine black hole will be detected by at least one agent. So from Lemma~\ref{lemma: reachCautiousStartNode} and Lemma~\ref{lemma: CautiousLeaderCautiousFollower} within at most $10n$ rounds after the first agent is destroyed by the Byzantine black hole there exists at least one agent that knows the exact location of the Byzantine black hole which can now explore the ring $R$ perpetually avoiding the Byzantine black hole. So if all three agents start from three different nodes the \textsc{PerpExploration-BBH} will be solved if each node has a whiteboard of memory $O(\log n)$. This proves the Theorem~\ref{thm: correct whitbrd scat}. We state the theorem again here for convenience.\\
  \textbf{Statement of theorem~\ref{thm: correct whitbrd scat}:} \textit{  Algorithm \textsc{PerpExplore-Scat-Whitbrd} solves\\ \textsc{PerpExploration-BBH} problem of a ring $R$ with $n$ nodes and with 3 synchronous agents initially scattered under the whiteboard model of communication}\\
  \subsection{Modification to the algorithm to include the cases where the initial starting nodes can have multiplicity}
  \label{Appendix: multiplicity whitbrd case}
  \begin{remark}
      \label{remark: Multiplicity case}
      Let $a_i,a_j$, and $a_k$ be three agents that start from two initial nodes, say $home_1$ and $home_2$. By Pigeon hole principle, exactly one of $home_1$ and $home_2$ initially must have two agents. Without loss of generality let $home_1$ have two agents, say $a_i$ and $a_k$, initially. In this case, the agents having multiplicity greater than one at the current node do not move. On the other hand the singleton agent, i.e., $a_j$ starting from $home_2$ moves clockwise marking each node with message \texttt{right}. If $a_j$ reaches $home_1$ before being destroyed by the Byzantine black hole, $home_1$ now has three agents co-located. Thus from here the agents execute the whiteboard version of \textsc{PerpExplore-Coloc-Pbl} (refer Subsection~\ref{subsection: colocwhitbrd simul}). On the other hand, let us consider the case when $a_j$ gets destroyed before reaching $home_1$. Note that irrespective of the location of $home_2$ it takes at most $n-1$ rounds for $a_j$ to reach $home_1$ from the beginning. So $a_i$ and $a_j$ waits for $n$ rounds and finds that no one has arrived yet. In this case, both $a_i$ and $a_k$ move to $home_2$ along the clockwise direction and start to perform the cautious walk, where the lowest ID agent among $a_i$ and $a_k$ changes its state to \textbf{Cautious-Leader}, whereas the other agent changes its state to \textbf{Cautious-Follower}. Next, the agent executing \textbf{Cautious-Leader} searches for the marking \texttt{right}. As argued earlier, within at most $4n$ rounds after an agent is destroyed, at least one of the remaining alive agents, detect the Byzantine black hole and continue to explore the ring avoiding that node.
  \end{remark}
  
\end{document}